\documentclass[journal]{IEEEtran}
\usepackage{cite}
\usepackage{subfigure}
\usepackage{amsthm}
\usepackage{amsmath}
\usepackage{amsfonts}
\usepackage{array}
\usepackage{algorithm}
\usepackage{algorithmic}
\usepackage{hhline}
\usepackage{multirow}
\usepackage{tabularx}
\usepackage{verbatim}
\usepackage{amsthm}
\usepackage[table,xcdraw]{xcolor}
\usepackage{mathabx}
\usepackage{hyperref}
\usepackage{bm}
\usepackage{kotex}
\usepackage{xcolor}
\usepackage{color}
\usepackage{booktabs}
\usepackage{tabulary}
\usepackage{color, colortbl}

\usepackage{graphics}
\usepackage{graphicx}
\usepackage{epstopdf}
\usepackage{epsfig,psfrag,epsf}

\theoremstyle{definition}
\newtheorem{remark}{Remark}
\newtheorem{thm}{Theorem}

\newtheorem{prob}{Problem}

\newtheorem{defi}{Definition}

\ifCLASSINFOpdf

\else

\fi

\definecolor{gray}{rgb}{0.5,0.5,0.5}
\def\tcr{\textcolor{red}}

\definecolor{islamicgreen}{rgb}{0.0, 0.56, 0.0}

\begin{document}

\title{A Reinforcement Learning Formulation of the Lyapunov Optimization:  Application to Edge Computing Systems with Queue Stability}

\author{Sohee Bae,~\IEEEmembership{Student Member,~IEEE},
Seungyul Han,~\IEEEmembership{Student Member,~IEEE},
and Youngchul Sung$^\dagger$,~\IEEEmembership{Senior Member,~IEEE} 
\thanks{This works was supported in part by Institute for Information \& communications Technology Promotion (IITP) grant funded by the Korea government (MSIT) (2019-0-00544, Building a simulation environment for machine learning-based edge computing operation optimization),     and in part by the National Research Foundation of Korea(NRF) grant funded by the Korea government. (MSIT) (2017R1E1A1A03070788)
 The authors are with the School of Electrical Engineering, KAIST, Daejeon, South Korea, 34141. Email: \{sh.bae, sy.han, ycsung\}@kaist.ac.kr.  $\dagger$ Corresponding author.
}}

\markboth{Submitted to IEEE Transactions on Networking}
{Shell \MakeLowercase{\textit{et al.}}: Bare Demo of IEEEtran.cls for Journals}

\maketitle

\begin{abstract}
In this paper, a deep reinforcement learning (DRL)-based approach to the Lyapunov optimization is considered  to minimize the time-average penalty while maintaining queue stability. A proper construction of state and action spaces is provided to form a proper Markov decision process (MDP) for the Lyapunov optimization. A condition for the reward function of reinforcement learning (RL) for queue stability is derived. Based on the analysis and practical RL with reward discounting, a class of reward functions is proposed for the DRL-based approach to the Lyapunov optimization.   The proposed DRL-based approach to the  Lyapunov optimization  does not required complicated optimization at each time step and operates with general non-convex and discontinuous penalty functions.   Hence, it  provides an alternative to the conventional drift-plus-penalty (DPP) algorithm  for the Lyapunov optimization. The proposed DRL-based approach is applied to resource allocation in edge computing systems with queue stability and numerical results demonstrate its successful operation.
\end{abstract}

\IEEEpeerreviewmaketitle

\section{Introduction}\label{sec:introduction}






The Lyapunov optimization in queueing networks is a well-known method to minimize a certain operating cost function while stabilizing queues in a  network \cite{neely2005, neely2010}. In order to stabilize the queues while minimizing the time average of the cost, the famous DPP algorithm minimizes the weighted sum of the drift and the penalty at each time step under the Lyapunov optimization framework. The DPP algorithm is widely used to jointly control the network stability and the penalty such as  power consumption in the traditional network and communication  fields \cite{tassiulas2006, neely2005, neely2010, network1,network2,network3,network4,network5,network6}. The Lyapunov optimization theorem guarantees that the DPP algorithm results in optimality within certain bound under some conditions. The Lyapunov optimization framework has been applied to many problems. For example, the backpressure routing algorithm  can be used for routing in multi-hop queueing networks \cite{tassiulas1990, tassiulas1993} and the DPP algorithm can be used for joint flow control and network routing \cite{neely2006energy, neely2005}.
In addition to these classical applications to conventional communication networks, the Lyapunov framework and the DPP algorithm for optimizing performance under queue stability can be applied to many optimization problems in  emerging systems such as energy harvesting and renewable energy such as smart grid and electric vehicles in which virtual queue techniques can be used to  represent the energy level with queue \cite{energyq1, energyq2, energyq3, energyq4, energyq5, energyq6, energyq7}.

Despite its versatility for the Lyapnov optimization, the DPP algorithm  is an instantaneous greedy algorithm and requires solving a  non-trivial  optimization problem for every time step. Solving  optimization for the DPP algorithm is not easy in case of complicated penalty functions. 
With the recent advances in DRL \cite{minh2015dqn}, RL has gained renewed interest in applications to many control problems for which classical RL not based on deep learning was not so effective \cite{drl_ra1, drl_ra2, drl_ra3, drl_ra4}. In this paper, we consider a DRL-based approach to the Lyapunov optimization  in order to provide an alternative to the DPP algorithm for time-average penalty minimization under queue stability. 
Basic RL is a MDP composed of a state space, an action space, a state transition probability and a policy.  The goal of RL is to learn a policy that maximizes the accumulated expected return \cite{sutton1998}.  The problem of time-average penalty minimization under queue stability 
can be formulated into an RL problem based on queue dynamics and  additive cost function.  The advantage of an RL-based approach is that RL exploits the trajectory of system evolution and does not require any optimization in the execution phase once the control policy is trained.  Furthermore, an RL-based approach can be applied to the case of complicated penalty functions with which numerical optimization at each time step for the DPP algorithm may be difficult.
Even with such advantages of an RL-based approach, an RL formulation for the Lyapunov optimization is not straightforward because of the condition of queue stability. 
The main challenge in an RL formulation of the Lyapunov optimization is how to incorporate the queue stability constraint 
into the RL formulation. Since the goal of RL is to maximize the expected accumulated reward, the desired control behavior is through the reward, and the success and effectiveness of the devised RL-based approach crucially depends on a well-designed reward function as well as good formulation of the state and action spaces.

\subsection{Contributions and Organization}
    
The contributions of this paper are as follows:

    

$\bullet$
We propose a proper MDP structure for the Lyapunov optimization by constructing the state and action spaces so that the formulation yields an MDP with a deterministic reward function, which facilitates learning.

$\bullet$  We propose a class of reward functions yielding queue stability as well as penalty minimization. The derived reward function is based on the relationship between the queue stability condition and the expected  accumulated reward which is the maximization goal of RL. The proposed reward function is in the form of one-step difference to be suited to practical RL which actually maximizes the discounted sum of rewards.

$\bullet$ Using the Soft-Actor Critic (SAC) algorithm \cite{sac}, we demonstrates that the DRL-based approach based on the constructed state and action spaces and the proposed reward function properly learns a policy that minimizes the penalty cost while maintaining queue stability.

$\bullet$ Considering the importance of edge computing systems in the  trend of network-centric computing \cite{201711.JSAC.Sun, 201808.ToN.Chen,2018006.TWC.Bi,201811.TMC.Neto,
201906.IoTJ.Chen2,2020EA.TMC.Huang,201812.TC.Dinh}, we applied the proposed DRL-based approach to the problem of resource allocation in  edge computing systems under queue stability, whereas many previous works investigated resource allocation in edge computing systems from different perspectives not involving queue stability at the edge server. The proposed approach provides a policy to optimal task offloading and self-computation to edge computing systems under task queue stability.

This paper is organized as follows.  In Section \ref{sec:SystemModel},  the system model is provided.  In Section \ref{sec:probstate}, the problem is formulated and the conventional  approach is explained. In Section \ref{sec:rlbased}, the proposed DRL-based approach to the Lyapunov optimization is explained. Implementation and experiments are provided in Sections \ref{sec:implement} and  \ref{sec:numerical}, respectively, followed by conclusion in Section \ref{sec:conclusion}.

\section{System Model}\label{sec:SystemModel}

In this paper, as an example of queuing network control, we consider an edge computing system  composed of an edge computing node, a cloud computing node and multiple mobile user  nodes, and consider the resource allocation problem at the edge computing node equipped with multiple queues. We will simply refer to the edge computing node and the cloud computing node as the edge node and the cloud node, respectively. We assume that there exist $N$ application types in the system, and multiple mobile nodes generate applications belonging to the $N$ application types and offload a certain amount of tasks to the edge  node. 
The edge  node  has $N$ task data queues,  one for each of the $N$ application types, and stores the upcoming tasks offloaded from the multiple mobile nodes according to their application types. Then,  the edge  node performs the tasks offloaded from the mobile nodes by itself or further offloads a  certain amount of tasks to the cloud  node through a communication link established between the edge  node and the cloud node. We assume that the maximum CPU processing clock rate of the edge node is $f_E$ cycles per second and the communication bandwidth between the edge  node and the cloud node is $B$ bits per second.     The considered overall system model is described in Fig. \ref{fig:system_model}.
    
\begin{figure}
        \centering
        \includegraphics[width=\linewidth]{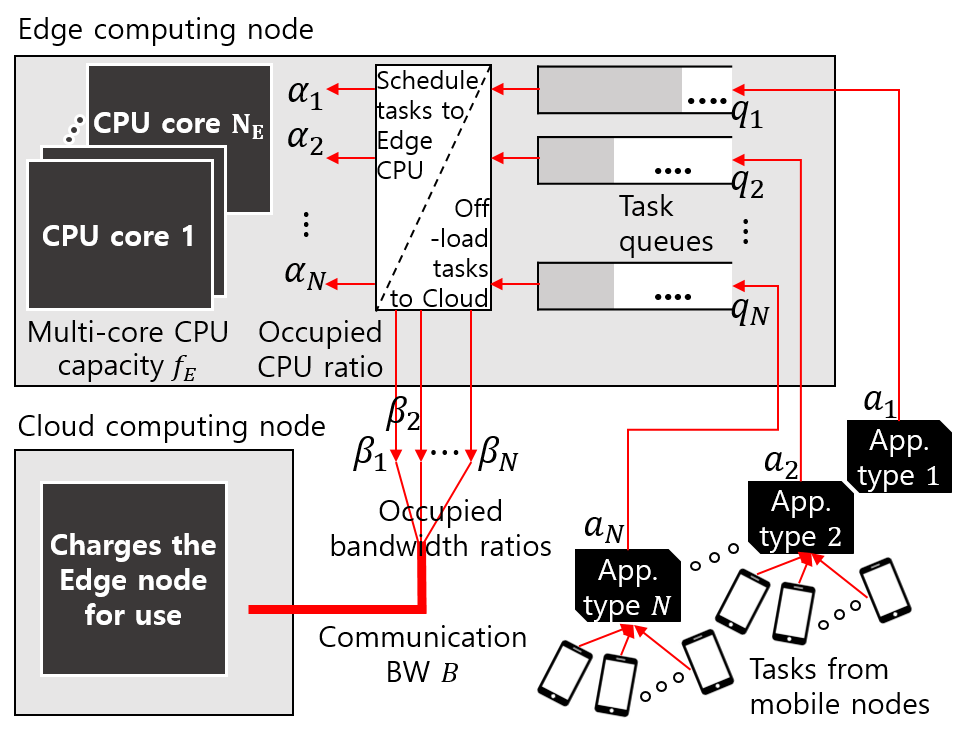}
        \caption{The edge-cloud system model with application queues}
        \label{fig:system_model}
\end{figure}

\subsection{Queue Dynamics}
\label{sec:QueueDynamics}

We assume that the arrival of the $i$-th application-type task at the $i$-th queue in the edge  node follows  a Poisson random process with arrival rate $\lambda_i$ [arrivals per second],  $i=1,\cdots,N$, and assume that the processing at the edge  node is  time-slotted with discrete-time slot index $t=0,1,2,\cdots$.
Let the task data bits offloaded from the mobile nodes to the $i$-th application-type data queue (or simply $i$-th queue) at time\footnote{Time is normalized so that one slot interval is one second  for simplicity.}  $t$ be denoted by $a_i(t)$, i.e.,  $a_i(t)$ is the sum of the arrived task data bits during the time interval $[t,t+1)$.  
We  assume that different application  type has  different work load, i.e., 
it requires a different number of CPU cycles to process one bit of task data for different application type, and  assume that the $i$-th application-type's one task bit requires $w_i$ CPU clock cycles for processing at the edge  node.

For each application-type task data in the corresponding queue, at time $t$ the edge  node determines the amount of allocated CPU processing resource for its own processing and the amount of offloading to the cloud  node while satisfying the imposed constraints. Let 
$\alpha_i(t)$ be the fraction of the edge-node CPU resource allocated  to the $i$-th  queue at the edge  node at time $t$ and let $\beta_i(t)$ be the fraction of the edge-cloud communication bandwidth for offloading  to the cloud  node for the $i$-th queue  at time $t$. 
Then, at the edge  node, $\alpha_i(t)f_E$ clock cycles per second are assigned to the $i$-th application-type task and $\beta_i(t)B$ bits per second for the $i$-th application-type task are aimed to be offloaded to the cloud  node at time $t$. 
Thus, the control action variables at the edge  node are given by
\begin{eqnarray}
 \pmb{\alpha}(t)&=&[\alpha_1(t),\cdots,\alpha_{N}(t)],  \label{eq:vecalpha}\\
    \pmb{\beta}(t)&=&[\beta_1(t),\cdots,\beta_{N}(t)],  \label{eq:vecbeta}
\end{eqnarray} 
where $\pmb{\alpha}(t)$ and $\pmb{\beta}(t)$ satisfy the following constraints:
\begin{equation}\label{eq:policy_constraint}
        \sum_{i=1}^{N} \alpha_i(t) \le 1, \qquad  \sum_{i=1}^{N} \beta_i(t) \le 1\qquad \text{ for all }t.
    \end{equation}
Our goal is to optimize the control variables $\pmb{\alpha}(t)$ and $\pmb{\beta}(t)$ under a certain criterion (which will be explained later) while satisfying the constraint \eqref{eq:policy_constraint}.

The queue dynamics at the edge node is then given by 
     \begin{equation} \label{eq:qdynamics1}
         q_i(t+1)=\left[q_i(t)+a_i(t)-\left(\underbrace{\frac{\alpha_i(t) f_E}{w_i}+\beta_i(t)B}_{=: b_i(t)}\right)\right]^+
     \end{equation}
where   $q_i(t)$ represents the length of the $i$-th queue at the edge  node at time $t$ for $i=1,\cdots,N$ and $[x]^+=\max(0,x)$. The second, third, and fourth terms in the right-hand side (RHS) of (\ref{eq:qdynamics1})  represent the new arrival, the reduction by edge node processing, and the offloading to the cloud node for the $i$-th  queue at time $t$, respectively. Here, the departure $b_i(t)$ at time $t$ is defined as  $\displaystyle b_i(t):= \frac{\alpha_i(t)f_E}{w_i}+\beta_i(t)B$. Note that by defining $b_i(t)$,  \eqref{eq:qdynamics1} reduces to a typical multi-queue dynamics model in queuing theory \cite{neely2010}. However,  in the considered edge-cloud system the departure occurs by two separate operations, computing and offloading, associated with $\pmb{\alpha}(t)$ and $\pmb{\beta}(t)$, and this makes the situation  more complicated.
The amount of actually offloaded task data bits from the edge node to  the cloud node for the $i$-th queue at time $t$ is given by
\begin{equation} \label{eq:aNpi}
 o_{i}(t)=\min\left(\beta_i(t)B, ~q_i(t)\!+\!a_i(t)\!-\!\frac{\alpha_i(t) f_E}{w_i}\right)
\end{equation}
because the remaining amount of task data bits at the $i$-th queue at the edge node can be less than the offloading target bits $\beta_i(t)B$.

For proper system operation, we require the considered edge computing system to be stable. Among several definitions of queuing network stability\cite{neely2010}, we adopt the following definition for stability:

\begin{defi}[{{Strong stability\cite{neely2010}}}] ~A queue  is strongly sta\label{def:ss} ble if
\begin{equation}
     {\lim\sup}_{t\rightarrow\infty}\frac{1}{t}\sum_{\tau=0}^{t-1}\mathbb{E}[q(\tau)]<\infty,
\end{equation}
         where $q(t)$ is the  length of the queue at time $t$.
         \end{defi}
\noindent We consider that the edge computing system  is stable if all the queues $q_i(t), i=1,\cdots,N$ in the edge node  are stable according to Definition \ref{def:ss}.

\subsection{Power Consumption Model and Cost Function}
\label{subsec:PowerConsumptionModel}

In order to model realistic computing environments, we assume multi-core computing at the edge  node, and assume that the edge  node has $N_E$ CPU cores with equal computing capability. We also assume that the Dynamic Voltage Frequency Scaling (DVFS) method is adopted at the edge node. DVFS is a widely-used technique for power consumption reduction, e.g.,  SpeedStep of Intel and Powernow of AMD.  DVFS adjusts   the CPU clock frequency and supply voltage  based on the  required CPU cycles per second to perform  given task in order to reduce power consumption \cite{mittal2014power}. That is, for a computationally easy task, the CPU clock frequency is lowered. On the other hand, for a computationally demanding task,  the CPU clock frequency is raised. Under our assumption that the edge CPU has maximum clock frequency of $f_E$ cycles per second and the edge CPU has $N_E$ CPU cores, each edge CPU core has maximum $f_E/N_E$ clock rate. 
Note from (\ref{eq:vecalpha}) and (\ref{eq:policy_constraint})
that the total assigned computing requirement for the edge  node at time $t$ is $f_E \cdot \sum_{i=1}^N \alpha_i(t)$. This total computing load is distributed to the $N_E$ CPU cores according to a multi-core workload distribution method. Hence, the assigned workload for the $j$-th core of the edge node is given by
\begin{equation}  \label{eq:f_core}
f_{E,j}= g\left(f_E,N_E,  \alpha_1(t), \cdots,\alpha_N(t)\right),    ~j=1,\cdots,N_E,
\end{equation}
where $g(\cdot)$ is the  multi-core workload distribution function of the edge CPU, and depends on individual design.

The power consumption at a CPU core consists mainly of two parts: the dynamic part and the static part \cite{CubicFunction}. We focus on the dynamic power consumption which is dominant as compared to the static part \cite{mangard2008power}. It is known that the dynamic power consumption is modeled as a cubic function of clock frequency, whereas the static part is modelled as a linear function of clock frequency \cite{voltage1, voltage2}.  With the focus on the dynamic part,  the power consumption at a CPU core can be modelled as  \cite{CubicPoor,CubicMao}
\begin{equation}  \label{eq:Pdynamic}
        P_{D} ~=~ \kappa f^3,
    \end{equation}
where $f$ is the CPU core clock rate and $\kappa$ is a constant depending on CPU implementation.  Then, the overall power consumption $C_E(t)$ at the edge  node can be  modelled as
\begin{align} \label{eq:edgeCost}
     C_E(t)&=\sum_{j=1}^{N_E}C_{E,j}(t)   
\end{align}
where $C_{E,j}$ denotes the power consumption at the $j$-th CPU core at the edge  node and is given by (\ref{eq:Pdynamic}) with the core operating clock rate $f$ substituted  by (\ref{eq:f_core}).
Note that for given $f_E$ and $N_E$, the power consumption $C_E(t)$ at time $t$ is a function of the control vector $\pmb{\alpha}(t)$, explicitly shown as
\begin{equation}  \label{eq:costEdget11}
C_E(t) = C_E(\alpha_1(t),\cdots,\alpha_N(t))    
\end{equation}
based on (\ref{eq:f_core}), (\ref{eq:Pdynamic}) and (\ref{eq:edgeCost}).

While $C_E(t)$ is the cost function measured in terms of the required power consumption for the edge  node caused by its own processing, we assume that the cloud  node charges cost $C_C(t)$ to the edge node based on  the amount of workload required to process the offloaded task bits to the cloud node $\sum_{i=1}^N w_i o_i(t)$, where $o_i(t)$ is given by \eqref{eq:aNpi}. 
Since $o_i(t)$ depends on $\pmb{\alpha}(t)$ and $\pmb{\beta}(t)$ as seen in (\ref{eq:aNpi}),  $C_C(t)$ as a function of the control variables is expressed as
\begin{equation} \label{eq:CostCloudt11}
C_C(t) = C_C(\alpha_1(t),\cdots,\alpha_N(t),\beta_1(t),\cdots,\beta_N(t)).   
\end{equation}
We assume that $C_C(t)$ is given in the unit of Watt under the assumption that power and monetary cost are interchangeable. We will use  $C_E(t)$ and $C_C(t)$ as the penalty cost in later sections. Table \ref{table:notations} summarizes the introduced notations.

    \begin{table}[t]
    \centering
    \caption{Summary of notations}\label{table:notations}
    \resizebox{0.48\textwidth}{!}{
        \begin{tabular}{|l|c|c|}
        \hline
        \rowcolor{lightgray}
            \textbf{name} & \textbf{stands for} & \textbf{unit} \\ \hline
            $N$ & Number of application types & -\\\hline
            $f_E$ & Maximum CPU clock rate of the edge  node &cycles/s\\ \hline
            $N_E$ & Number of CPU cores at the edge  node & -\\\hline
            $B$ &   \begin{tabular}[c]{@{}c@{}} Communication bandwidth\\from the edge node to the cloud node\end{tabular}&bits/s\\\hline
            $\lambda_i$ & \begin{tabular}[c]{@{}c@{}}Poisson arrival rate of the $i$-th app. type\\ at the edge  node \end{tabular}& arrivals/s\\\hline
            $f_{E,j}$ & \begin{tabular}[c]{@{}c@{}}
            Assigned workload for the $j$-th CPU core\\ 
            at the edge  node\end{tabular} &cycles/s\\ \hline
            $a_i(t)$ & Arrival task bits of the $i$-th application type at time $t$ &bits\\\hline
            $b_i(t)$ & Departure task bits of the $i$-th application type at time $t$ &bits\\\hline
            $o_i(t)$ & Offloaded task bits of the $i$-th application type at time $t$ &bits\\\hline
            $w_i$&  Workload for the $i$-th application type &cycles/bit\\\hline
            $q_i(t)$ & \begin{tabular}[c]{@{}c@{}} Queue length at time $t$ for the $i$-th queue\\
             at the edge node \end{tabular} & bits\\\hline
            
            $\alpha_i(t)$ &
            \begin{tabular}[c]{@{}c@{}} Edge CPU resource allocation factor for \\
            the $i$-th queue at time $t$\end{tabular}& - \\\hline
            $\beta_i(t)$ & 
            \begin{tabular}[c]{@{}c@{}} Communication bandwidth allocation factor for the $i$-th queue\\
            from the edge  node to the cloud  node  at time $t$\end{tabular}& -\\\hline
            $C_E(t)$ & Cost for computing  at the edge  node at time $t$ & \begin{tabular}[c]{@{}l@{}}Watt \\
            \end{tabular}\\\hline
            $C_{E,j}(t)$ & 
             \begin{tabular}[c]{@{}c@{}} Cost for computing  at the $j$-th CPU core \\
            at the  edge  node at time $t$\end{tabular} & \begin{tabular}[c]{@{}c@{}}Watt \\
            \end{tabular}\\\hline
            $C_C(t)$ & Cost for offloading to the cloud  node at time $t$ & \begin{tabular}[c]{@{}c@{}} Watt \\
            \end{tabular}\\\hline
        \end{tabular}
        }
    \end{table}

\section{Problem Statement and Conventional Approach}\label{sec:probstate}

In this section, based on the derivation in Section \ref{sec:SystemModel} we formulate the problem of optimal resource allocation at the edge  node under queue stability. Since the arrival process is random, the optimization cost is random. Hence, we consider the minimization of the time-averaged expected cost while maintaining  queue stability for stable system operation.  The considered optimization problem is formulated as follows:

\begin{prob}\label{prob:original_prob}
\begin{align}
         &\displaystyle\min_{\pmb{\alpha}(t), \pmb{\beta}(t)} \lim_{T\rightarrow \infty}\frac{1}{T}\sum_{t=0}^{T-1}\mathbb{E}[C_E(t)+C_C(t)]  \label{eq:prob1Cost}\\
         \mbox{s.t.}\qquad&\limsup_{t\rightarrow\infty}\frac{1}{t}\sum_{\tau=0}^{t-1}\mathbb{E}[q_i(\tau)]<\infty~\text{ for all }i\label{eq:strongstable}\\
         &\displaystyle\sum_{i=1}^N\alpha_i(t)\le1\text{ and }\sum_{i=1}^N\beta_i(t)\le1\text{ for all }t,
\end{align}
where  $\pmb{\alpha}(t)$ and $\pmb{\beta}(t)$ are defined in  \eqref{eq:vecalpha} and \eqref{eq:vecbeta}, respectively, $C_E(t)$ and $C_C(t)$ are the cost functions defined in \eqref{eq:costEdget11} and \eqref{eq:CostCloudt11}, respectively, and $q_i(t)$ is the length of the $i$-th queue at time $t$.
\end{prob}

\noindent Note that (\ref{eq:strongstable}) implies that we require all the queues in the system are strongly stable as a constraint for optimization. 
It is not easy to solve Problem \ref{prob:original_prob} directly. 
A conventional approach to Problem \ref{prob:original_prob} is based on the Lyapunov optimization \cite{neely2010}.  The Lyapunov optimization defines the quadratic Lyapunov function and the Lyapunov drift  as follows \cite{neely2010}:
     \begin{align}
        L(t)&=\frac{1}{2}\sum_{i=1}^{N} q_i(t)^2 \label{eq:LyFunc}\\
        \Delta L(t)&=L(t+1)-L(t). \label{eq:LyDrift}
     \end{align}
 It is known that Problem 1 is feasible when the average service rate (i.e., average departure rate) is strictly larger than the average arrival rate \cite{neely2010}, i.e., $
 \lambda_i\mu_i - \frac{\bar{\alpha}_i f_E}{w_i}-\bar{\beta}_i B <-\epsilon$  
    for some $\epsilon>0$ and some constants $\bar{\alpha}_i$ and $\bar{\beta}_i$, $i=1,\cdots,N$. Here, $\mu_i$ is the average task packet size at each arrival at the $i$-th application queue.
When  stable control is feasible, we want to determine 
the instantaneous service rates $\alpha_i(t)$ and $\beta_i(t)$ at each time for cost-efficient stable control of the system.  A widely-considered conventional method to determine the instantaneous service rates for Problem 1 is the DPP algorithm. The DPP algorithm minimizes the DPP   instead of the penalty (i.e., cost) alone, given by \cite{neely2010}
    \begin{equation} \label{eq:dppexpression}    
        \Delta L(t)+V[C_E(t)+C_C(t)]
\end{equation}
for a positive weighting factor $V$ which determines the trade-off between the drift and the penalty. 
In  \eqref{eq:dppexpression}, the original queue stability constraint is absorbed  as the drift term in an implicit manner.     
The DPP in our case is expressed  as 
        \begin{align}
            \Delta L(t)&+V [C_E(t)+C_C(t)] \nonumber\\
            \le& ~\frac{1}{2} \sum_{i=1}^N\left(a_i(t)-\frac{\alpha_i(t)f_E}{w_i}-\beta_i(t)B\right)^2\nonumber\\
            &+\sum_{i=1}^Nq_i(t)\left(a_i(t)-\frac{\alpha_i(t)f_E}{w_i}-\beta_i(t)B\right)\nonumber\\
                       &+V\left[C_E(\pmb{\alpha}(t))+C_C(\pmb{\alpha}(t),\pmb{\beta}(t))\right],
            \label{eq:DpPupperBound}
        \end{align}
where the inequality is due to ignoring the operation $[x]^+$ in \eqref{eq:qdynamics1}. 
The basic DPP algorithm minimizes the DPP expression \eqref{eq:DpPupperBound} in a greedy manner, which is summarized in
Algorithm  \ref{alg:mdpp}\cite{neely2010}.
\begin{algorithm}[t] 
		\caption{Basic  Drift-Plus-Penalty Algorithm \cite{neely2010}}
		\begin{algorithmic}[1]
			\label{alg:mdpp}
			\STATE \textbf{Initialization:} \\
				Set  $t = 1$.
			\STATE \textbf{Repeat:} \\
				\begin{enumerate}
				\item[1)] Observe $\bm a(t)=[a_1(t),\cdots,a_N(t)]$ and $\bm q(t)=[q_1(t),\cdots,q_N(t)]$.
				\item[2)] Choose actions $\pmb{\alpha}(t)$ and $\pmb{\beta}(t)$ to minimize \eqref{eq:DpPupperBound}.
				\item[3)] Update $\bm q (t)$ according to (\ref{eq:qdynamics1}) and $t \leftarrow t+1$.
				\end{enumerate}
		\end{algorithmic}
\end{algorithm}
It is known that under some conditions this simple greedy DPP algorithm  yields a solution that satisfies strong stability for all queues and its resultant time-averaged penalty is within some constant bound from the optimal value of Problem 1 \cite{neely2010}.
Note that there exist two terms generated from $\Delta L(t)$ in the RHS of  \eqref{eq:DpPupperBound}: one is the square of the difference between the arrival and departure rates and the other is the product of the queue length and the difference between the arrival and departure rates. In many cases, the quadratic term in the RHS of  \eqref{eq:DpPupperBound} is replaced by a constant upper bound based on certain assumptions on the arrival and departure rates $a_i(t)$ and $b_i(t)$, and only the second term $q_i(t)(a_i(t)-b_i(t))$ is considered   \cite{neely2010}.

\section{The Proposed Reinforcement Learning-Based Approach}\label{sec:rlbased}

Although Problem \ref{prob:original_prob} can be approached  by the conventional DPP algorithm. The DPP algorithm has several disadvantages: It is an instantaneous greedy optimization and requires solving an optimization problem at each time step, and numerical optimization with a complicated penalty function can be difficult.  As an alternative, in this section, we consider a  DRL-based approach to Problem \ref{prob:original_prob}, which exploits the trajectory of system evolution and does not require any optimization in the execution phase once the control policy is trained.

A basic RL is an MDP composed of a state space $\mathcal{S}$, an action space $\mathcal{A}$, a state transition probability $P: \mathcal{S}\times \mathcal{A} \rightarrow \mathcal{S}$, a reward function $r: \mathcal{S} \times \mathcal{A} \rightarrow \mathbb{R}$, and a policy $\pi: \mathcal{S} \rightarrow \mathcal{A}$ \cite{sutton1998}. At time $t$,  the agent which has policy $\pi$ observes a state $s_t \in \mathcal{S}$ of the environment and performs an action $\mathrm{a}_t$   according to the policy, i.e., $\mathrm{a}_t \in \mathcal{A} \sim \pi(\mathrm{a}_t|s_t)$. Then, a reward $r_t$,  depending on the current state  $s_t$ and the agent's action $\mathrm{a}_t$, is given to the agent according to the reward function $r_t
=r(s_t,\mathrm{a}_t)$, and the state of the environment changes to a next state $s_{t+1}$ according to the state transition probability, i.e., $s_{t+1} \sim P(s_{t+1}|s_t,\mathrm{a}_t)$.   The goal is to learn a policy to  maximize the accumulated expected return.

Problem \ref{prob:original_prob} can be formulated into a RL problem based on the queue dynamics  \eqref{eq:qdynamics1} and the additive cost function \eqref{eq:prob1Cost}, in which  the agent is the resource allocator of the edge  node and tries to learn an optimal policy for resource allocation at the edge  node while stabilizing the queues.  Since we do not assume the knowledge of the state transition probability $P$, our approach belongs to model-free RL \cite{sutton1998}.  The main challenge in the RL formulation of Problem \ref{prob:original_prob} is how to enforce the queue stability constraint in \eqref{eq:strongstable} 
into the RL formulation, and  
the success and effectiveness of the devised RL-based approach depends critically on a well-designed reward function as well as good formulation of the state and action spaces.

\subsection{State and Action Spaces}\label{sec:states}

 In order to define the state and action spaces, we clarify the operation in time domain. Fig. \ref{fig:RLstate_trans} describes our timing diagram for RL operation. For causality under our definition of the state and the action, we assume one time step delay for overall operation.
 Recall that  time is normalized in this paper. Hence, the discrete time index $t$ used in the queue dynamics in Section \ref{sec:QueueDynamics} can also mean the continuous time instant $t$.    As seen in Fig. \ref{fig:RLstate_trans}, considering the actual timing relationship, we define the quantities as follows. $q_i(t)$ is the length of the $i$-th queue at the continuous time instant $t$, $a_i(t)$ is the sum of the arrived task bits for the continuous time interval $[t,t+1)$, and $b_i(t)$ is the serviced task bits during the continuous time interval $[t+1,t+2)$ based on the observation of $q_i(t)+a_i(t)$. The one step delayed service $b_i(t)$ is incorporated in computing the queue length $q_i(t+1)$ at the time instant $t+1$ due to the assumption of one step delayed operation for causality. 
Then, the state variables that we determine at the RL discrete time index $t$ for the considered edge system are as follows.

\begin{enumerate}

    \item The main state variable is the queue length  including the arrival $a_i(t)$: $q_i(t)+a_i(t)$,  $i=1,\cdots,N$.
    
    \item Additionally, we include  the queue length   at time $t$  before the arrival $a_i(t)$, i.e., $q_i(t)$ or 
the arrival $a_i(t)$ itself in the set of state variables.
    
    \item The workload for each application type: $w_i$, $i=1,\cdots,N$.
    
    \item The actual CPU use factor\footnote{The nominal use of the edge-node CPU for the $i$-th queue at time $t$ is $\alpha_i(t)f_E$. However, when the amount of task data bits in the $i$-th queue is less than the value $\alpha_i(t)f_E/w_i$, the actual CPU use for the $i$-th queue by action at time $t$  denoted as $\tilde{\alpha}_i(t)f_E$
   is less than $\alpha_i(t)f_E$. In the queue dynamics  \eqref{eq:qdynamics1},
this effect is handled by  the function $[x]^+=\max(0,x)$.} for the $i$-th queue  at time $t-1$: $\tilde{\alpha}_i(t)$, $i=1,\cdots,N$.

    \item The required CPU cycles at the cloud node for the offloaded tasks at time $t$: $\sum_{i=1}^N w_i o_i(t)$.

    \item The time average of $a_i(t)$ for the most recent 100 time slots: $\frac{1}{100}\sum_{\tau=t-99}^{t} a_i(\tau)$, $i=1,\cdots,N$.
    
 \end{enumerate}

The action variables of the policy are $\alpha_i(t)$ and $\beta_i(t)$, $i=1,2,\cdots,N$, and the constraints on the actions are $\sum_{i=1}^N \alpha_i(t)\le 1$ and $\sum_{i=1}^N \beta_i(t) \le 1$.


\begin{figure}[!t]
\centering
    \includegraphics[scale=0.45]{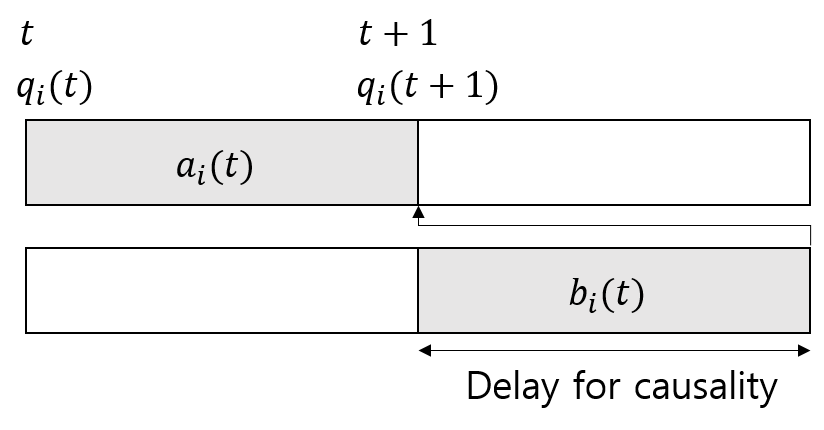}
\caption{Timing diagram}
        \label{fig:RLstate_trans}
\end{figure}

\begin{remark} \label{remark:Timetdef} We can view the last moment of the continuous time interval $[t,t+1)$ as the reference time for the RL discrete time index $t$. 
Note that $q_i(t)+a_i(t)$ is set as the main state variable at RL discrete time index $t$, and the service (i.e., action) $b_i(t)$ is based on  $q_i(t)+a_i(t)$. 
In this case, from the RL update viewpoint, the transition is in the form of 
\begin{equation}
\begin{array}{ll}
\text{Current state} & s_t =\{q_i(t)+a_i(t),\cdots\}  \\
\text{Current action} &  \mathrm{a}_t =(\pmb{\alpha}(t),\pmb{\beta}(t)) ~\Rightarrow~ \text{determines}~b_i(t)   \ \\
\text{Current reward} & r_t(s_t,\mathrm{a}_t) ~\text{as a function of}~ (s_t,\mathrm{a}_t)  \\
\text{Queue  update} & q_i(t+1) = [q_i(t) + a_i(t) - b_i(t)]^+  \\
\text{Next state}  &s_{t+1}=\{q_i(t+1)+a_i(t+1),\cdots\}.
\end{array}
\end{equation}
This definition of the state and timing  is crucial for proper RL training and operation. The reason why we define the reference time and the state in this way will be explained in Section  \ref{subsubsec:RDQSPM}.
\end{remark}

\subsection{Reward Design for Queue Stability and Penalty Minimization}
\label{subsubsec:RDQSPM}

Suppose that we simply use the negative of the DPP expression as the reward function for RL, i.e, 
\begin{equation}
        r_t^{guess} = -[ \Delta L(t) + V (C_E(t)+C_C(t))  ],
\end{equation}
where the drift term  $\Delta L(t)$ is given by
    \begin{align}  \label{eq:DelLtINReDesign}
     &\Delta L(t)=\frac{1}{2}\sum_{i=1}^N \left[ q_i(t+1)^2 - q_i(t)^2  \right] 
     \end{align}
and try to maximize the  accumulated sum $\sum_t r_t^{guess}$     by RL. Then, does RL with this reward function lead  to the behavior that we want?  In the following, we derive a proper reward function for RL to solve Problem \ref{prob:original_prob} and answer the above question in a progressive manner, which is the main contribution of this paper.

The main point for an RL-based approach to Problem  \ref{prob:original_prob} to yield queue stability is to exploit the fact  
that the goal of RL is to maximize the accumulated reward (not to perform a greedy optimization at each time step) and hence achieving  the intended behavior is through a well-designed reward function. In order to design such a reward function for RL to yield queue stability, we start with the 
following result:

\begin{thm} \label{thm:rewardCond}
Suppose that $q_i(0)=0, \forall i$ (we will assume initial zero backlog for all queues in the rest of this paper) and  the reward $r_t$ at time $t$ for RL satisfies the following condition: \begin{equation} \label{eq:RtDesignCond1}
        r_t \le U - \eta \sum_{i=1}^N q_i(t+1)
\end{equation}
for some finite constant $U$  and finite positive $\eta$. Then, RL trained with such $r_t$ tries to strongly stabilize the queues. Furthermore, if the following condition is satisfied
\begin{equation} \label{eq:RtDesignCond2}
    r_{min} \le r_t, ~~~\forall t
    \end{equation}
     for some finite $r_{min}$ in addition to 
\eqref{eq:RtDesignCond1}.
Then, the resulting queues by RL with such $r_t$ are strongly stable.
\end{thm}

\begin{proof}
   Taking expectation on both sides of 
   \eqref{eq:RtDesignCond1}, we have
    \begin{equation}  \label{eq:RtCondProof1}
        \mathbb{E}[ r_\tau] \le U - \eta \sum_{i=1}^N \mathbb{E}[q_i(\tau+1)].
    \end{equation}
   Summing \eqref{eq:RtCondProof1} over $\tau=0,1,\cdots, t-1$, we have
        \begin{equation} \label{eq:RtCondProof2}
        \sum_{\tau=0}^{t-1} \mathbb{E}[ r_\tau] \le Ut - \eta \sum_{\tau=0}^{t-1}\sum_{i=1}^N \mathbb{E}[q_i(\tau+1)].
        \end{equation}
Rearranging \eqref{eq:RtCondProof2} and dividing by $\eta t$, we have
        \begin{equation}  \label{eq:RtCondProof3}
        \frac{t+1}{t}\frac{1}{t+1} \sum_{\tau=0}^{t}\sum_{i=1}^N \mathbb{E}[q_i(\tau)] \le \frac{U}{\eta} - \frac{1}{\eta}\frac{1}{t} \sum_{\tau=0}^{t-1} \mathbb{E}[ r_\tau],
        \end{equation}
where $\frac{t+1}{t} \rightarrow 1$ as $t$ increases, and we used $q_i(0)=0$. Note that the left-hand side (LHS) of \eqref{eq:RtCondProof3} is the time-average of the expected sum queue length.
Since the goal of RL is to maximize the accumulated expected reward $\sum_{\tau=0}^{t-1} \mathbb{E}[ r_\tau]$, RL with $r_t$ satisfying \eqref{eq:RtDesignCond1}
tries to stabilize the queues by making the average queue length small. That is, for the same $t$, when $\sum_{\tau=0}^{t-1}\mathbb{E}[r_\tau]$ is larger, the average queue length becomes smaller.

Furthermore, if the condition \eqref{eq:RtDesignCond2} is satisfied in addition, the second term in the RHS of  \eqref{eq:RtCondProof3} is upper bounded as 
$-\frac{1}{t} \sum_{\tau=0}^{t-1} \mathbb{E}[ r_\tau] \le -r_{min}$. Hence, from \eqref{eq:RtCondProof3}, we have
        \[
        \frac{t+1}{t}\frac{1}{t+1} \sum_{\tau=0}^{t}\sum_{i=1}^N \mathbb{E}[q_i(\tau)] \le \frac{U}{\eta} - \frac{r_{min}}{\eta}= \frac{1}{\eta}(U-r_{min}).
        \]
If the sum of the queue lengths is bounded, the length of each queue is bounded. Therefore, in this case,  the queues are strongly stable by Definition \ref{def:ss}. (Note that $r_{min} < U$ from \eqref{eq:RtDesignCond1}.)
\end{proof}

\noindent  Note that from the definition of strong stability in Definition \ref{def:ss} and the fact that RL tries to maximize the accumulated expected reward, the condition  \eqref{eq:RtDesignCond1}  and resultant \eqref{eq:RtCondProof3} can be considered as a natural starting point for reward function design.

With the guidance by  Theorem \ref{thm:rewardCond}, we design a  reward function for RL to learn a policy to simultaneously decrease the average queue length and the penalty, i.e., the resource cost. For this we set the RL reward function as the sum of two terms:
$r_t = r_t^Q + r_t^P$,
where $r_t^Q$ is the queue-stabilizing part and $r_t^P$ is the penalty part given by $r_t^P = - V[C_E(t)+C_C(t)]$ with a weighting factor $V$. Based on  Theorem  \ref{thm:rewardCond}, 
we  consider the following class of functions as a candidate for the queue-stabilizing part $r_t^Q$: 
      \begin{equation}\label{eq:queueRewards}
        r_t^{Q} = - \rho \sum_{i=1}^{N} [q_i(t+1)]^\nu,~~~\nu \ge 1
    \end{equation} 
with some positive constant $\rho$. 
Then, the total reward at time $t$ for RL is given by
   \begin{equation}  \label{eq:ProposedFinalReward}
       r_t = r_t^Q + V r_t^P = -\rho \sum_{i=1}^N [q_i(t+1)]^\nu - V [C_E(t) + C_C(t)].
   \end{equation}
The property of the reward function \eqref{eq:ProposedFinalReward} is provided in the following theorem:
        
\begin{thm}\label{thm:RewardVerification}
The queue-stabilizing part $r_t^Q$ of the reward function \eqref{eq:ProposedFinalReward} makes RL with the reward \eqref{eq:ProposedFinalReward} try to strongly stabilize the queues.
\end{thm}

\begin{proof}  
 Note from Section 
    \ref{subsec:PowerConsumptionModel} that     $C_E(t)=\sum_{j=1}^{N_E}C_{E,j}(t)$ with $C_{E,j}=\kappa f_{E,j}^3$ and $C_C(t) = C_C(\sum_{i=1}^N w_i o_i(t))$, where the $j$-th edge CPU core clock frequency $f_{E,j}$ and the offloading $o_i(t)$ from the $i$-th queue to the cloud node are given by \eqref{eq:f_core} and \eqref{eq:aNpi}, respectively. We have $C_E(t) \ge 0 $ and $C_{C}(t) \ge 0$ by design.
    Furthermore, we have
    \begin{equation}
    f_{E,j} \le \frac{f_E}{N_E}, ~\forall j~~\mbox{and}~~
    o_i(t) \le B, ~\forall i
    \end{equation}
    by considering the full computational and communication resources. Hence, we have
    \begin{equation}
        C_E(t) \le \kappa \frac{f_E^3}{N_E^2} ~~\mbox{and}~~ C_C(t) \le C_C\left(B \sum_{i=1}^N w_i\right),
    \end{equation}
    with slight abuse of the notation $C_C$ as a function of the offloaded task bits in the RHS of the second inequality.     Therefore, we have
    \begin{equation}  \label{eq:CeCeUB}
     -\kappa \frac{f_E^3}{N_E^2} -C_C\left(B \sum_{i=1}^N w_i\right) \le    -[C_E(t)+C_{C}(t)] \le 0. 
    \end{equation}
    Now we can upper bound  $r_t$ as follows:
    \begin{align}
        r_t &=  -\rho \sum_{i=1}^N [q_i(t+1)]^\nu - V [C_E(t) + C_{C}(t)]   \\
            &\stackrel{(a)}{\le}  -\rho \sum_{i=1}^N [q_i(t+1)]^\nu,  \label{eq:theo2UB2}  \\
            &\stackrel{(b)}{\le}  -\rho \sum_{i=1}^N q_i(t+1)+\rho N,  \label{eq:proofTheo2_35}
    \end{align}
where Step (a) is valid due to $-[C_E(t)+C_{C}(t)]\le 0$ and Step (b) is valid due to the inequality $x^\nu \ge x -1$
for $\nu \ge 1$.  Then, by setting $U=\rho N$ and $\eta=\rho$, we can apply the result of Theorem  \ref{thm:rewardCond}.
\end{proof}

\begin{remark}
Note that in the proof of Theorem \ref{thm:RewardVerification},  the upper bound \eqref{eq:theo2UB2} becomes tight when $V$ is zero.  Thus, the queue-stabilizing part of the reward dominantly operates when $V$ is small so that $-V[C_E(t)+C_C(t)]\approx 0$. In case of large $V$, a small reduction in $C_E(t)+C_C(t)$ yields a large positive gain in $-V[C_E(t)+C_C(t)]$ and thus $r_t^P$ becomes dominant. Even in this case,  the queue-stabilizing part $r_t^Q$ itself still operates towards the direction of queue length reduction due to the negative sign in front of the queue length term. This is what Theorem  \ref{thm:RewardVerification} means.  However, in case of large $V$, more reward can be obtained by saving $C_E(t)+C_C(t)$ while increasing $q_i(t+1)$, and  
 the reward  \eqref{eq:ProposedFinalReward} does not guarantee strong queue stability since it  is not lower bounded as  \eqref{eq:RtDesignCond2} due to the structure of $-\rho\sum_i [q_i(t+1)]^\nu$. Hence, a balanced $V$ is required for 
simultaneous queue stability and penalty reduction. 
\end{remark}

Now let us investigate the reward function \eqref{eq:ProposedFinalReward} further. First, note that the penalty part $r_t^P=-V[C_E(t)+C_C(t)]$ is a deterministic function of action $\pmb{\alpha}_i(t)$ and $\pmb{\beta}_i(t)$.  Second, consider the term $q_i(t+1)$ in $r_t^Q$ in detail. $q_i(t+1)$ is decomposed as 
    \begin{eqnarray}\label{eq:qab}
    q_i(t+1) 
         &=&\underbrace{q_i(0)+\sum_{\tau=0}^{t-1}[a_i(\tau)-b_i(\tau)]}_{=q_i(t)}+ a_i(t) - b_i(t)   \nonumber\\
          &=&  \underbrace{q_i(t)+a_i(t)}_{\text{state at time}~t}-\underbrace{b_i(t)}_{\text{action at time}~t}  \label{eq:qitp1decomp}
    \end{eqnarray}
under the assumption    of  $q_i(t)+a_i(t)\ge b_i(t)$  for simplicity. 
Note that $b_i(\tau)$ is a deterministic  function of the action $\pmb{\alpha}(\tau)$ and $\pmb{\beta}(\tau)$ for $\tau=0,1,\cdots,t$ but the arrivals $a_i(\tau)$, $\tau=0,1,\cdots,t$ are random quantities  uncontrollable by the policy $\pi$.  
Recall that the reward function in RL is a function of state and action in general. In the field of RL, it is known that an environment with probabilistic reward is more difficult to learn than an environment with deterministic reward \cite{wang2020reinforcement}.  That is, for a given state, the agent performs an action and receives a reward depending on the state and the action. When the received reward is probabilistic especially with large variance, it is difficult for the agent to know whether the action is good or bad for the given state.  Now, it is clear why we defined $q_{i}(t)+a_i(t)$ as a state variable at time $t$, as mentioned in Remark \ref{remark:Timetdef}, and defined the timing structure as defined in Section \ref{sec:states}.
By defining $q_{i}(t)+a_i(t)$ as a state variable, the reward-determining quantity $q_i(t+1)$  becomes a deterministic function of the state and the action as seen in \eqref{eq:qitp1decomp}, and the random arrivals $a_i(\tau)$, $\tau=0,1,\cdots, t$ are absorbed in the state. In this case, the randomness caused by $a_i(t+1)$ is in the state transition:
\begin{align}
    s_t &= \{q_i(t)+a_i(t), \cdots \}  \nonumber\\
    s_{t+1} &= \{\underbrace{q_i(t) + a_i(t)}_{\in~ s_i(t)}- \underbrace{b_i(t)}_{\text{action}} + 
    \underbrace{a_i(t+1)}_{\text{random term}},\cdots\}.
\end{align}
That is, the next state follows $s_{t+1} \sim P(s_{t+1}|s_t,\mathrm{a}_t)$ and the distribution of the random arrival $a_i(t+1)$ affects the state transition probability $P$. Note that in this case the transition is Markovian  since the arrival $a_i(t+1)$  is independent of the arrivals at other time slots.  Thus, the whole set up falls into an MDP with a deterministic reward function.  However,  if we had defined the state at time $t$ as $q_i(t)$ instead of $q_i(t)+a_i(t)$ (this setup does not require one time step delay for causality), then $q_i(t+1)$ would have been decomposed as 
\begin{equation}
    q_i(t+1) 
          =  \underbrace{q_i(t)}_{\text{state at time $t$}} + \underbrace{a_i(t)}_{\text{random term}}  -\underbrace{b_i(t)}_{\text{action at time $t$}}
\end{equation}
to yield a random probabilistic reward, and this would have made learning difficult.

Although RL with  the reward function $r_t^Q=-\rho\sum_{i=1}^N [q_i(t+1)]^\nu$ with $\nu\ge 1$ added to the penalty part  $r_t^P$ tries to strongly stabilize  the queues by Theorem \ref{thm:RewardVerification} through the relationship  \eqref{eq:RtCondProof3}, we want to reshape the queue-stabilizing part  $r_t^Q$ of the reward into a {\em discounted form} to be suited to practical RL, while maintaining the reward sum equivalence needed for queue length control by RL  through the relationship \eqref{eq:RtCondProof3}. 
Our reward reshaping is based on the fact that  training in RL is typically based on episodes, which is assumed here too. Let $T$ be the length of each episode. Then,  under the assumption of $q_i(0)=0$,   we can express the accumulated reward  over one episode as 
\begin{align}
\frac{1}{T}\sum_{t=0}^{T-1} r_{t}^{Q} &=-\rho \frac{1}{T}\sum_{t=0}^{T-1} \sum_{i=1}^N[q_i(t+1)]^\nu  \label{eq:reshape39}\\
&\stackrel{(a)}{=} -\rho \sum_{t=0}^{T-1}\sum_{i=1}^N \frac{T-t}{T} [q_i(t+1)^\nu - q_i(t)^\nu]\label{eq:RewReshapDiff}
\end{align}
where the equality $(a)$   is valid because the coefficient in front of the term $[q_i(\ell)]^\nu$ for each $\ell$ in \eqref{eq:RewReshapDiff}  is given by  
\begin{equation}  \label{eq:discountEqual}
\frac{T-(\ell-1)}{T} - \frac{T-\ell}{T}=\frac{1}{T}.
\end{equation}
Thus, by defining
\begin{equation}  \label{eq:tildertq}
    \tilde{r}_t^Q = - \rho  \sum_{i=1}^N \frac{T-t}{T} [q_i(t+1)^\nu - q_i(t)^\nu], ~~~\nu \ge 1,\\
\end{equation}
we have the sum equivalence between the original reward $r_t^Q$ in \eqref{eq:queueRewards} and the reshaped reward $\tilde{r}_t^Q$ except the factor $1/T$, as seen in  \eqref{eq:RewReshapDiff}. 
Since $r_t^Q$ satisfies $r_t^Q \le U - \eta \sum_{i=1}^N q_i(t+1)$ as seen in the proof of Theorem 2 and $\tilde{r}_t$ satisfies  $\sum_{\tau=0}^{t-1} \tilde{r}_\tau^Q =\frac{1}{t} \sum_{\tau=0}^{t-1}r_\tau^Q$ due to \eqref{eq:RewReshapDiff},  by summing the first condition over time $0,1,\cdots,t-1$ and using the second condition, we have  
\begin{equation}  \label{eq:controlReshaped}
\sum_{\tau=0}^{t-1} \tilde{r}_\tau=\frac{1}{t}\sum_{\tau=0}^{t-1} r_\tau \le U - \eta \frac{1}{t} \sum_{\tau=0}^{t-1} \sum_{i=1}^N q_i(t+1).
\end{equation}
Rearranging the terms in   \eqref{eq:controlReshaped} and taking expectation, we have 
\begin{equation}  \label{eq:tildeRtControl}
        \frac{t+1}{t}\frac{1}{t+1} \sum_{\tau=0}^{t}\sum_{i=1}^N \mathbb{E}[q_i(\tau)] \le \frac{U}{\eta} - \frac{1}{\eta} \sum_{\tau=0}^{t-1} \mathbb{E}[ \tilde{r}_\tau^Q].
\end{equation}
Hence,  we can still control the queue lengths by RL  with the reshaped reward $\tilde{r}_t^Q$. As compared to \eqref{eq:RtCondProof3},  the factor $1/t$ in front of the sum reward term in \eqref{eq:RtCondProof3} disappears in  \eqref{eq:tildeRtControl}.
The key aspect of  the reshaped reward is that the reward at time $t$ is discounted by the factor $\frac{T-t}{T}$, which is a monotone decreasing function of $t$ and decreases from one to zero as time elapses. This fact makes the reshaped reward suitable for practical RL  and this will be explained shortly.

\begin{remark} 
Note that the reshaped reward in \eqref{eq:tildertq} can be rewritten as
   \begin{equation}   \label{eq:remark3rtildeQ}
    \tilde{r}_t^Q = -\rho \sum_{i=1}^{N} \frac{T-t}{T}  [(\underbrace{q_i(t) + a_i(t)}_{\text{state at} ~t}-\underbrace{b_i(t)}_{\text{action at}~t})^\nu - \underbrace{q_i(t)^\nu}]. 
\end{equation}
Again, we want to express $\tilde{r}_t^Q$ as a deterministic function of the state and the action.  The term $q_i(t)+a_i(t)$ is already included in the set of state variables and $b_i(t)$ is deterministically dependent on the action. For our purpose, the last term $q_i(t)$ in the RHS of \eqref{eq:remark3rtildeQ} should be deterministically determined by the state. Hence, 
we included either $q_i(t)$ or $a_i(t)$ in addition to $q_i(t)+a_i(t)$ in the set of state variables, as seen in  Section \ref{sec:states}. 
In the case of $a_i(t)$ as a state variable,  $q_i(t)$ in  \eqref{eq:remark3rtildeQ}  is determined with no uncertainty from the state variables $q_i(t)+a_i(t)$ and $a_i(t)$. 
\end{remark}

Finally, let us consider practical RL.     
Practical RL tries to minimize the sum of exponentially discounted rewards $\sum_t \gamma^t r_t$ with a discount factor $\gamma < 1$  not $\sum_t r_t$ in order to guarantee the convergence of the Bellman equation \cite{sutton1998,khan2012reinforcement}, whereas our derivation up to now assumed the minimization of the sum of rewards. 
In RL theory, the Bellman operator is typically used to estimate the state-action value function, and it becomes a contraction mapping when the rewards are exponentially discounted by a discount factor $\gamma <1$. Then, the state-action value function converges to a fixed point and proper learning is achieved \cite{khan2012reinforcement}.  Hence, this discounting should be incorporated in our reward design. Note that $\gamma^t$ monotonically decreases from one to zero as time goes and that our reshaped reward $\tilde{r}_t^Q$ has the internal discount factor $\frac{T-t}{T}$, which also monotonically decreasing from one to zero as time goes.   
Even though the two discount factors are not the same exactly, their monotone decreasing behavior matches and plays the same role of discount.  With  the existence of the external RL discount factor $\gamma~(<0)$, we redefine our  reward for RL aiming at queue stability and penalty minimization as 
\begin{equation} \label{eq:finalrewardFunction}
 \hat{r}_t = - \rho  \sum_{i=1}^N  [q_i(t+1)^\nu - q_i(t)^\nu]  -V[C_E(t)+C_C(t)],
\end{equation}
with $\nu \ge 1$. Then, with the external RL discount factor, the actual  queue stability-related part of the reward in practical RL becomes $-\rho \sum_{i=1}^N \gamma^t[q_i(t+1)^\nu - q_i(t)^\nu] $.  Our reward \eqref{eq:tildertq} tries to approximate this actual reward by a first-order approximation 
with one step time difference form $[q_i(t+1)^\nu - q_i(t)^\nu]$.
Thus, the queue lengths can be controlled through the relationship \eqref{eq:tildeRtControl} by directly maximizing the sum of discounted rewards by RL.  Note that the penalty part is also discounted when we use the reward \eqref{eq:finalrewardFunction} in practical RL with reward discounting. However, this is not directly related to queue length control and such discounting is typical in practical RL.

\begin{remark}  \label{rem:remark4}
Note that the RHS of \eqref{eq:reshape39} is the time average of $q_i(t+1)^\nu$ over time 0 to $T-1$ with equal weight $1/T$, whereas the RHS of \eqref{eq:RewReshapDiff} is the time average of one-step difference $[q_i(t+1)^\nu - q_i(t)^\nu]$ over time 0 to $T-1$ with unequal discounted weight $\frac{T-t}{T}$.  Note that  the one-step difference form makes the impact of each $q_i(t+1)$ equal in the discounted average, as seen in \eqref{eq:discountEqual}. Suppose that we directly use $r_t^Q=-\rho \sum_{i=1}^N q_i(t+1)^\nu$ without one-step difference reshaping for practical RL. Then, the queue-stability-related part in the sum of discounted rewards $\sum_{t=0}^{T-1}\gamma^t r_t$ in practical RL becomes
\begin{equation}
\sum_{t=0}^{T-1} \gamma^t r_t^Q = -\rho \sum_{t=0}^{T-1} \sum_{i}\gamma^tq_i(t+1)^\nu, ~~~0< \gamma< 1,
\end{equation}
where  $\sum_{t=0,1,\cdots} \gamma^t (\cdot)$ can be viewed as a weighted time average with some scaling. 
Thus, the queue length of the initial phase of each episode is overly weighted.  Reshaping into the one-step difference form mitigates this effect by trying to make the impact of each $q_i(t+1)$ equal in the discounted average within first-order linear approximation.
\end{remark}

\begin{remark}
Now, suppose that we maximize the sum of undiscounted rewards and  use the one-step discounted reward, i.e., $\sum_{t=0}^{T-1} \hat{r}_t^Q$.
 Then, we have
\begin{align*}
    \frac{1}{\rho}\sum_{t=0}^{T-1} \hat{r}_t^Q &= -\sum_{t=0}^{T-1} \sum_i [q_i(t+1)^\nu -q_i(t)^\nu] \nonumber\\
    &=     -\sum_i q_i(T)^\nu + \sum_i q_i(T-1)^\nu-\sum_i q_i(T-1)^\nu
    +\nonumber\\
    &~~~\cdots+\sum_iq_i(0)^\nu=-\sum_i q_i(T)^\nu+\sum_i\underbrace{q_i(0)^\nu}_{=0}.
\end{align*}
Hence, the time average of queue length required to implement the strong stability  in Definition \ref{def:ss}
does not appear in the reward sum and 
maximizing the sum reward tries to minimize the queue length only at the final time step.  Thus, the one-step difference reward form  \eqref{eq:finalrewardFunction}  is valid for RL minimizing the sum of discounted rewards.
\end{remark}

When $\nu =2$, the queue-stabilizing part $\hat{r}_{t}^{Q}$  of our reward  \eqref{eq:finalrewardFunction}  reduces to  the negative of the drift term  $-\Delta L(t)$ in the Lyapunov framework, given in \eqref{eq:DelLtINReDesign}.  So, the negative of DPP can be used as the reward for practical RL minimizing the sum of discounted rewards not for RL minimizing the sum of rewards.  When $\nu=1$, $\hat{r}_{t}^{Q}$  simply reduces to 
$\hat{r}_{t}^{Q} = -\rho \sum_{i=1}^{N}[a_i(t)-b_i(t)]$.

Considering that RL tries to maximize the expected accumulated reward and  $\mathbb{E}[a_i(t)]=\lambda_i$, 
we can further stabilize the reward by replacing the random arrival $a_i(t)$ with its mean $\lambda_i$. For example, when $\nu=1$, we use
\begin{equation} \label{eq:hatrtiqxxx}
\hat{\hat{r}}_{t}^{Q} = -\rho \sum_{i=1}^{N} [ \lambda_i - b_i(t)],
\end{equation}
and when $\nu=2$, we use
\begin{equation} \label{eq:hatrtiqxxxNu2}
\hat{\hat{r}}_{t}^{Q} = -\rho \sum_{i=1}^{N}  \left\{2q_i(t)[ \lambda_i - b_i(t)] + [\lambda_i - b_i(t)]^2\right\},
\end{equation}
where  the arrival rate $\lambda_i$ can easily be estimated.  Note that with $\nu = 1$, the second upper bounding step  (b) in \eqref{eq:proofTheo2_35} is not required and hence we have a tighter upper bound, whereas the drift case $\nu=2$ has the advantage of length balancing across the queues due to the property of a quadratic function.

\section{Implementation}  \label{sec:implement}

Among several popular recent {DRL} algorithms, we choose the Soft Actor-Critic (SAC) algorithm, which is a state-of-the-art algorithm optimizing the policy in an off-policy manner \cite{sac,dac}. 
SAC  maximizes the discounted sum of the expected  return and the policy entropy to enhance exploration. Thus, the SAC policy objective function is given by 
\begin{equation}\label{eq:sac_objective}
J(\pi) = \mathbb{E}_{\xi\sim\pi} \left[    \sum_{t=0}^{T-1}\gamma^t({r}_t+\zeta \mathcal{H}(\pi(\cdot|s_t)))
        \right],
\end{equation}
where $\pi$ is the policy, $\xi=(s_0,\mathrm{a}_0,s_1,\mathrm{a}_1,\cdots)$ is the state-action trajectory, $\gamma$ is the reward discount factor, $\zeta$ is the entropy weighting factor,  $\mathcal{H}(\pi)$ is the policy entropy, and ${r}_t$ is the reward.
    \begin{figure}
        \centering
        \includegraphics[width=\linewidth]{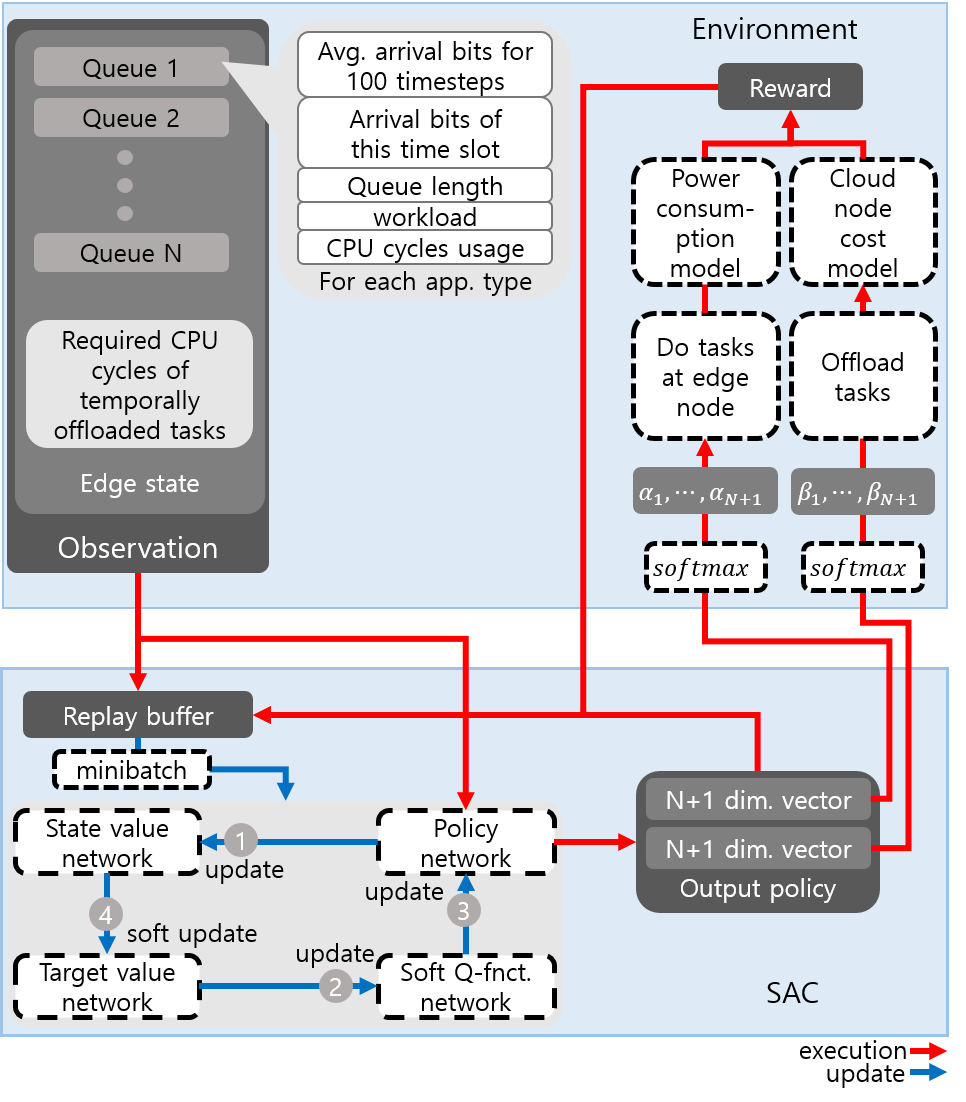}
        \caption{Process diagram of the environment and SAC networks}
        \label{fig:env_sac_flow}
    \end{figure}
   \begin{table}[!th]
        \centering
        \begin{tabular}{|l|c|}
        \hline
        \rowcolor{lightgray}
            \textbf{Parameter} & \textbf{Value}\\
            Optimizer & Adam optimizer\\
            Learning rate & $3\cdot 10^{-4}$\\
            Discount factor $\gamma$ & $0.999$ \\
            Replay buffer size & $10^6$ \\
            Number of hidden layers & $2$ \\
            Number of hidden units per layer & $256$ \\
            Number of samples per minibatch & $256$ \\
            Nonlinearity & ReLU\\
            Target smoothing coefficient  & 0.005\\
            Target update interval & 1\\
            Gradient step & 1\\
            \hline
        \end{tabular}
        \caption{SAC Hyperparameters}
        \label{tab:sac_params}
    \end{table} 
Fig. \ref{fig:env_sac_flow} describes the overall diagram of our edge computing environment and the SAC agent.  
For implementation of SAC, we followed the standard algorithm in \cite{sac} with the state variables and the action variables defined  in Section \ref{sec:states}. 
In order to implement the condition \eqref{eq:policy_constraint}, we implemented the policy deep neural network with dimension $2N+2$ by adding two  dummy variables for $\alpha_{N+1}(t) ~(\ge 0)$ and $\beta_{N+1}(t) ~(\ge 0)$ and applied the softmax function at the output of the policy network satisfying 
\begin{equation}\label{action_constraint}
        \sum_{i=1}^{N+1} \alpha_i(t) = 1, ~~~ \sum_{i=1}^{N+1} \beta_i(t) = 1, ~~\forall t.
\end{equation}   
Then, we took only $\alpha_i(t)$ and $\beta_i(t)$, $i=1,2,\cdots,N$ from the neural network output layer.
The used hyperparameters are the same as those in \cite{sac} except the discount factor $\gamma$ and the values are provided in Table \ref{tab:sac_params}. We assumed that  the DRL policy at the edge node updated its status and performed its action at every second. The episode length for learning was $T=5000$ time steps. Our implementation source code is available at  Github \cite{SoheeGit}.

\section{Experiments}  \label{sec:numerical}

\subsection{Environment Setup}
\label{subsec:EnvSetup}

In order to test the proposed DRL-based approach, we considered the following system.  With heavy computational load on smartphones caused by artificial intelligence (AI) applications, we  considered AI applications to be offloaded from smartphones to the edge node. 
The considered three AI application types were speech recognition, natural language processing and face recognition. 
The number of required CPU cycles $w_i$ for each application type was roughly estimated by running open or commercial software on smart phones and personal computers.  We assumed that the arrival process of the $i$-th application-type tasks was a Poisson process with mean arrival rate $\lambda_i$ [arrivals/second], as mentioned in Section \ref{sec:QueueDynamics}. We further assumed that the data size $d_i$ [bits] of each task arrival of the $i$-th application type followed   
  a truncated normal distribution $N_T(\mu_i, \sigma_i, d_{i,min},d_{i,max})$. We first set   the minimum and maximum data sizes $d_{i,min}$ and $d_{i,max}$ of one task for  the $i$-th application type  and then set the mean and standard deviation as $ \mu_i = (d_{i,max}+d_{i,min})/{2}$ and   $\sigma_i =  (d_{i,max}-d_{i,min})/{4}$. 
 \begin{table}[h]
    \centering
    \caption{Parameter setup for each application type}
        \begin{tabular}{|l|c|c|c|c|c|c|}
        \hline
        \rowcolor{lightgray} 
       \begin{tabular}[c]{@{}l@{}}\textbf{Application} \\ \textbf{type}\end{tabular}
        & $w_i$ &  \textbf{Distribution of $d_i$}   &$\lambda_i$ \\ \hline
        \begin{tabular}[c]{@{}l@{}}
        \textbf{Speech}\\ \textbf{recognition}
        \end{tabular}  & 10435 & $N_T(170,130,40,300)$ (kB)  & 5 \\ \hline
        \begin{tabular}[c]{@{}l@{}}
        \textbf{Natural language}\\ \textbf{processing}
        \end{tabular}
         & 25346 & $N_T(52,48,4,100)$ (kB) & 8 \\ \hline
        \begin{tabular}[c]{@{}l@{}}\textbf{Face}\\ \textbf{Recognition}\end{tabular}
         & 45043  & $N_T(55,45,10,100)$ (kB) &4 \\ \hline
        \end{tabular}
    \label{tab:app_info}
    \end{table}
We set the average number of  task arrivals of the $i$-th application-type $\lambda_i$ and   the minimum and maximum data sizes $d_{i,min}$ and $d_{i,max}$ of one task for the $i$-th application-type as shown in Table \ref{tab:app_info}.

We assumed a scenario in which  the cloud  node had  larger processing capability than the edge node and a good portion of processing was done at the cloud node.  We assumed that the edge node had 10 CPU cores and each CPU core had the processing capability of 4 Giga cycles per second [Gcycles/s or simply GHz]. Hence, the total processing power of the edge node was 40 Gcycles/s.  
Among the valid range of $\nu \ge 1$, we used $\nu=1$ or $\nu=2$, i.e.,  \eqref{eq:hatrtiqxxx} or
 \eqref{eq:hatrtiqxxxNu2} for our reward function. Thus, the overall reward function was given by 
\begin{equation}  \label{eq:SimReward}
  \hat{\hat{r}}_t^Q -V[C_E(t)+C_C(t)],
\end{equation}
where  $C_E(t)$ was the cost of edge processing given in
\eqref{eq:edgeCost} as
\begin{equation}  \label{eq:CEtFormula}
    C_E(t) = \sum_{j=1}^{N_E}\kappa f_{E,j}^3
\end{equation}
and $C_C(t)$ was the cost of offloading from the edge node to the cloud node. 
We considered two  cases for $C_C(t)$. The first one was a simple continuous function given by 
\begin{equation}  \label{eq:CctFormula}
C_C(t) =   \kappa  N_C \left( \frac{\sum_{i=1}^N w_i o_i(t)}{N_C} \right)^3,
\end{equation}
where $o_i(t)$ was the number of the offloaded task bits of the $i$-th application type to the cloud node, given by \eqref{eq:aNpi}, and $N_C$ was the number of the CPU cores at the cloud node. Note that the cost function \eqref{eq:CctFormula} follows the same principle as  \eqref{eq:CEtFormula} under the assumption that the overall workload $\sum_i w_i o_i(t)$ offloaded to the cloud node is evenly distributed over the $N_C$ CPU cores at the cloud node. We set the number of CPU cores at the cloud node as $N_C=54$ with each core having 4 Gcycles/s processing capability. Thus, the maximum cloud processing capability was set as 216 Gcycles/s.
 The second cost function for $C_C(t)$ was a discontinuous function, which will be explained in Section  \ref{subsubsec:discontinuousRewFunc}.
  We set  $\kappa=\displaystyle\frac{1}{(400 \text{GHz})^3}$ based on a rough estimation\footnote{Suppose that a  CPU core of $4$GHz clock rate  consumes $35$W and that  $10$kWh = 36,000 kW$\cdot$ s costs one dollar.   Then, from $1 :36,000 \text{kW}\cdot \text{s}=\kappa f_{E,j}^3 :35 \text{W}\cdot \text{s}$, we obtain $\kappa = \frac{35}{36,000\cdot 10^3 \cdot (4\text{GHz})^3}\approx 1/(400\textbf{GHz})^3$. }  and  $\rho=10^{-9}$ in our simulations. In fact, the values of $\kappa$ and $\rho$  were not critical since we  swept the weighting factor $V$ between the delay-related term and the penalty cost in order to see the  overall trade-off. This value setting was for the numerical dynamic range of the used SAC code.

{\em Feasibility Check:}  Note that the average arrival rates in terms of CPU clock cycles per second for the three application types are given by
\begin{equation*}
    \begin{aligned}
        \lambda_1\cdot \mu_1\cdot w_1&=5 \cdot 170 \cdot 8\cdot 1024   \cdot 10435 \simeq 72.7  ~\text{Gcycles/s}\\
        \lambda_2\cdot \mu_2\cdot w_2&=8 \cdot 52 \cdot 8 \cdot 1024 \cdot  25346 \simeq 86.4 \text{Gcycles/s}\\
        \lambda_3\cdot \mu_3\cdot w_3&=4 \cdot 55 \cdot 8 \cdot 1024 \cdot 45043 \simeq 81.2 \text{Gcycles/s}.
    \end{aligned}
\end{equation*}
The sum of the above three rates is roughly $240$ Gcycles/s among which $40$ Gcycles/s can be processed at maximum at the edge node. Table     \ref{tab:costcomp} shows the average values of $C_E(t)$ and $C_C(t)$ for different offloading to the cloud node based on \eqref{eq:CctFormula} under the assumption that the assigned workload is evenly distributed over the CPU cores both at the edge and cloud nodes 
(the unit of the first two columns  in Table     \ref{tab:costcomp} is Gcycles/s and the unit of the remaining three columns is $G^3\kappa = 10^{21}\kappa$). 
 \begin{table}[h]
    \centering
    \caption{Environment Setup Check}
        \begin{tabular}{|c|c|c|c|c|}
        \hline
        \rowcolor{lightgray} 
        At Edge  & At Cloud  & $C_E(t)$ & $C_C(t)$ & $C_E(t)+C_C(t)$ \\ 
        \hline
        40  & 200 &  640 & 2743  & 3383 \\
        \hline
        30  & 210 &  270  &  3175 & 3445\\
        \hline
        20  & 220 & 80  & 3651  & 3731 \\
        \hline
        \end{tabular}
    \label{tab:costcomp}
    \end{table}
As seen in Table     \ref{tab:costcomp}, the edge node should process for smaller overall cost.  If we offload all tasks to the cloud node, the required communication bandwidth is given by $\sum_i \lambda_i \mu_i = 12.2 \text{Mbps}$. So, we set the communication bandwidth $B=20$Mbps so that the communication is not a bottleneck for system operation. Since the overall service rate provided by the edge and cloud nodes  is $256 ~(=40+216)$ Gcycles/s and the average arrival rate is $240$ Gcycles/s and the communication bandwidth is not a bottleneck, the overall system is feasible to control.  The DRL resource allocator should learn a policy that distributes the arriving tasks to the edge and cloud nodes optimally.

\subsection{Convergence and Comparison with the DPP Algorithm}
\label{subsec:simDPPcomp}

We tested our DRL-based approach with the proposed reward function for the system described in 
Section \ref{subsec:EnvSetup} 
with $C_C(t)$ given by the simple continuous cost function    \eqref{eq:CctFormula}, and compared its performance to that of the DPP algorithm. For comparison, we used the basic DPP algorithm in Algorithm  \ref{alg:mdpp} with the cost at each time step given by 
\begin{align} 
    & \sum_{i=1}^N q_i(t) \left(a_i(t)-\frac{\alpha_i(t)f_E}{w_i} - \beta_i(t) B\right)  \nonumber \\
    &~~~~~~~~~+ V'[C_E(\pmb{\alpha}(t)) + C_C(\pmb{\alpha}(t),\pmb{\beta}(t))], \label{eq:DPPsimReward}
\end{align}
where the quadratic term in the RHS of \eqref{eq:DpPupperBound} was replaced by constant upper bound and omitted.  Note that the weighting factor $V'$ is different from  the weighting factor $V$ in  \eqref{eq:SimReward} in order to take into account the difference in the stability-related terms in the two cost functions \eqref{eq:SimReward} and \eqref{eq:DPPsimReward}.  In the DPP algorithm, for each time step, optimal $\pmb{\alpha}(t)$ and $\pmb{\beta}(t)$ were found by minimizing \eqref{eq:DPPsimReward} for given $q_i(t)$ and $a_i(t)$. For this numerical optimization, we used sequential quadratic programming (SQP), which is an iterative method solving the original constrained nonlinear optimization with successive  quadratic approximations \cite{nocedal2006sequential} and  is  widely used with several available software including MATLAB, LabVIEW and SciPy. We used SciPy of python to implement the DPP algorithm.

  \begin{figure}[t]
        \subfigure[]{
        \centering
        \includegraphics[width=\linewidth]{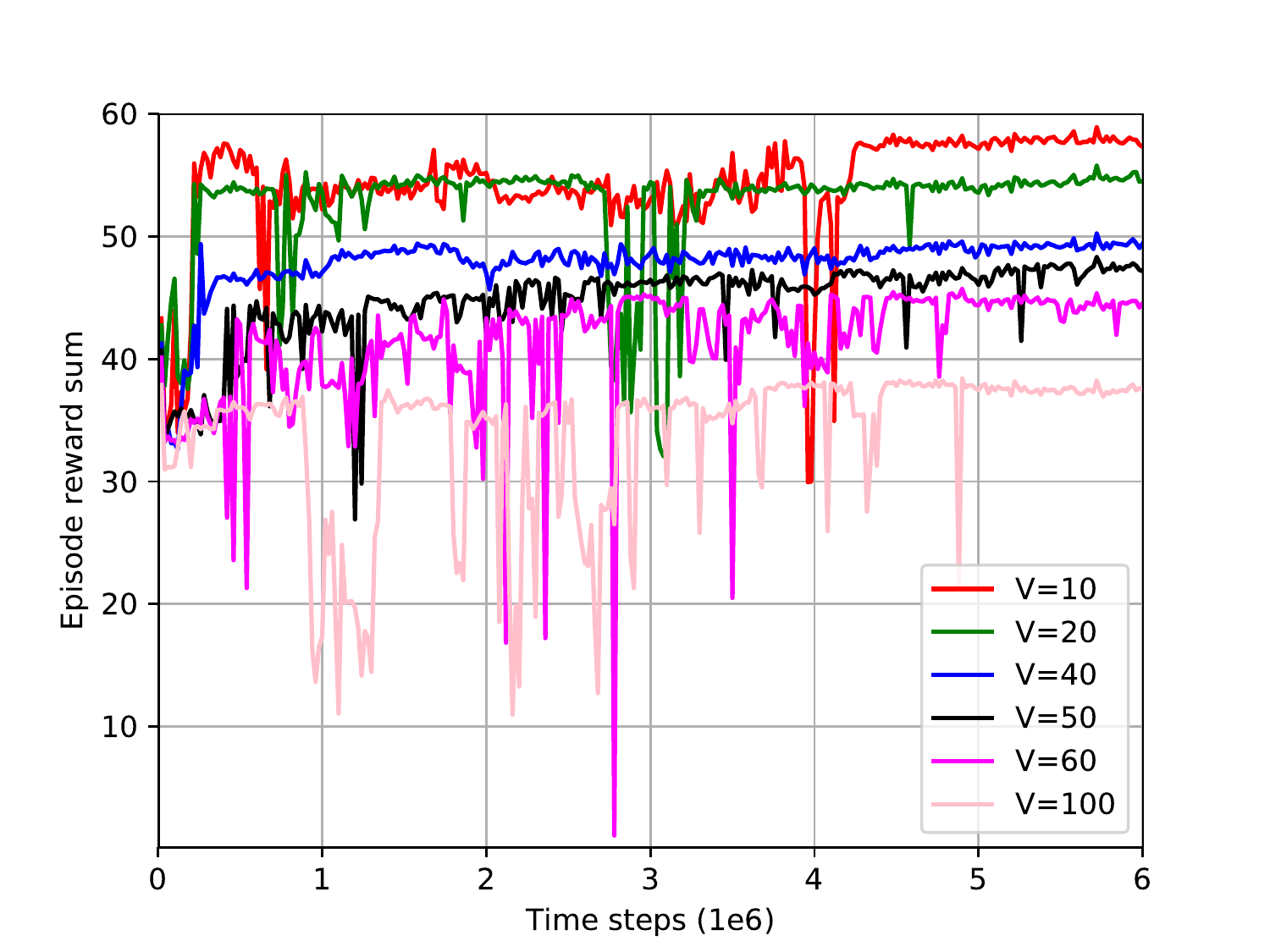}
        }
        \subfigure[]{
        \centering
        \includegraphics[width=\linewidth]{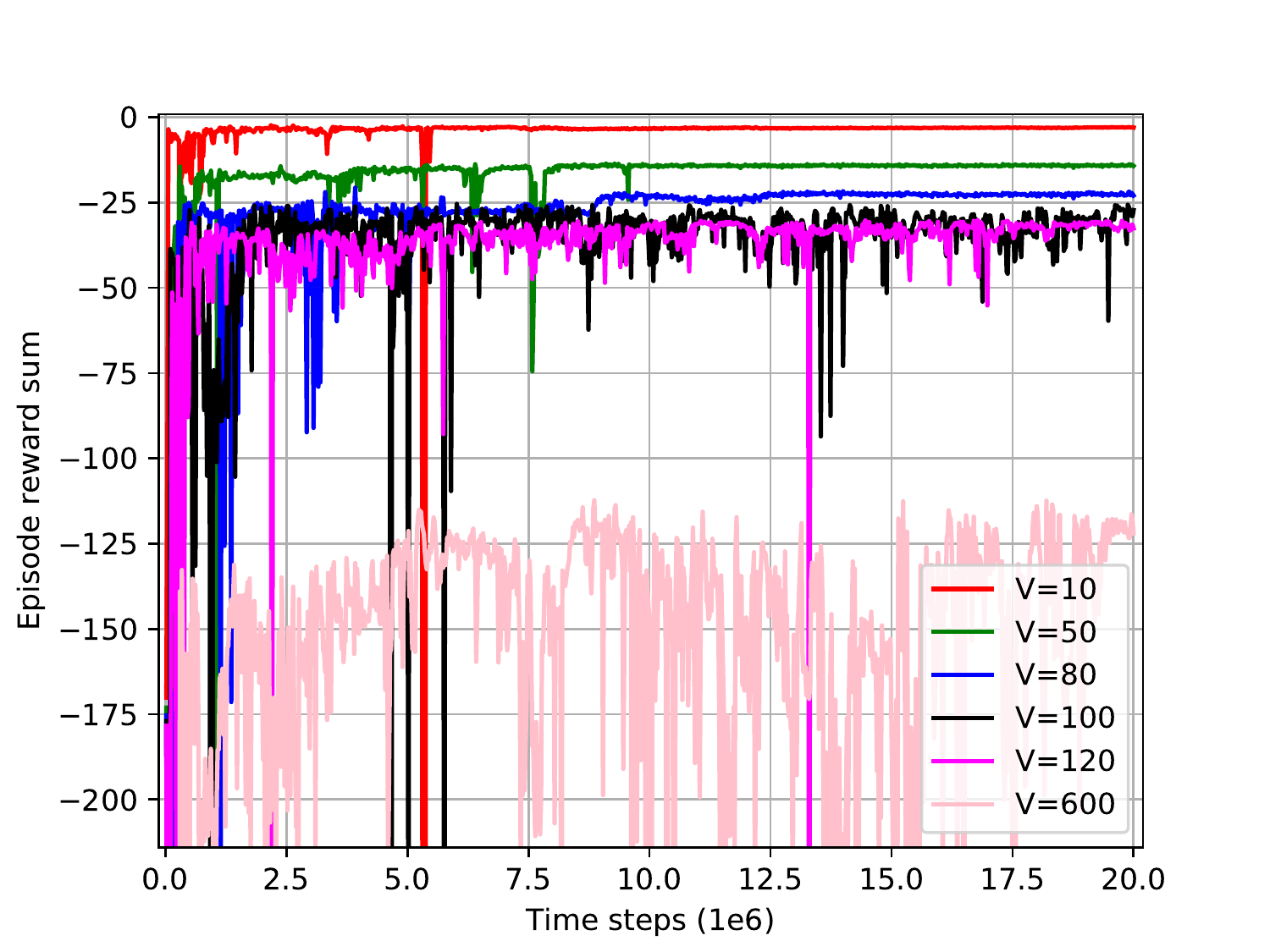}
    }
    \caption{The learning curve of the proposed DRL algorithm: (a) $\nu=1$, i.e., \eqref{eq:SimReward} with \eqref{eq:hatrtiqxxx}  and (b) $\nu=2$, i.e., \eqref{eq:SimReward} with \eqref{eq:hatrtiqxxxNu2}}
    \label{fig:learningCurves_step11}
    \end{figure}

 For SAC, we did the following. In the beginning, all weights in the neural networks for the value function and the policy were randomly initialized.   We generated four episodes, collected 20,000 samples, and stored them  into the sample buffer, where one episode for training was composed of $T=5000$ time steps and in the beginning of each episode, all queues were emptied.\footnote{In Atari games, one episode typically corresponds to one game starting from the beginning.} Then, with the samples in the sample buffer, we trained the neural networks. With this trained policy, we generated one episode  and evaluated the performance with the evaluation-purpose episode. Then, we again generated and stored four episodes of 20,000 samples into the sample buffer,  trained the policy with the samples in the sample buffer, and evaluated the newly trained policy with one evaluation episode. We repeated this process.

 First, we checked the convergence of  SAC with the proposed reward function.  Fig.  \ref{fig:learningCurves_step11} shows its learning curve for different values of the weighting factor $V$ in \eqref{eq:SimReward}. The $x$-axis in Fig.  \ref{fig:learningCurves_step11}  is the training episode time step (not including the evaluation episode time steps) and the $y$-axis is the episode reward sum $\sum_{\tau=0}^{T-1} r_\tau$ based on \eqref{eq:SimReward} without discounting for the evaluation episode corresponding the $x$-axis value. Note that although SAC itself tries to maximizes the sum of discounted rewards, we plotted the undiscounted episode reward sum by storing \eqref{eq:SimReward} at each time step.  It is observed that the proposed DRL algorithm converges as time goes. 
   \begin{figure}[t]
        \centering
        \includegraphics[width=\linewidth]{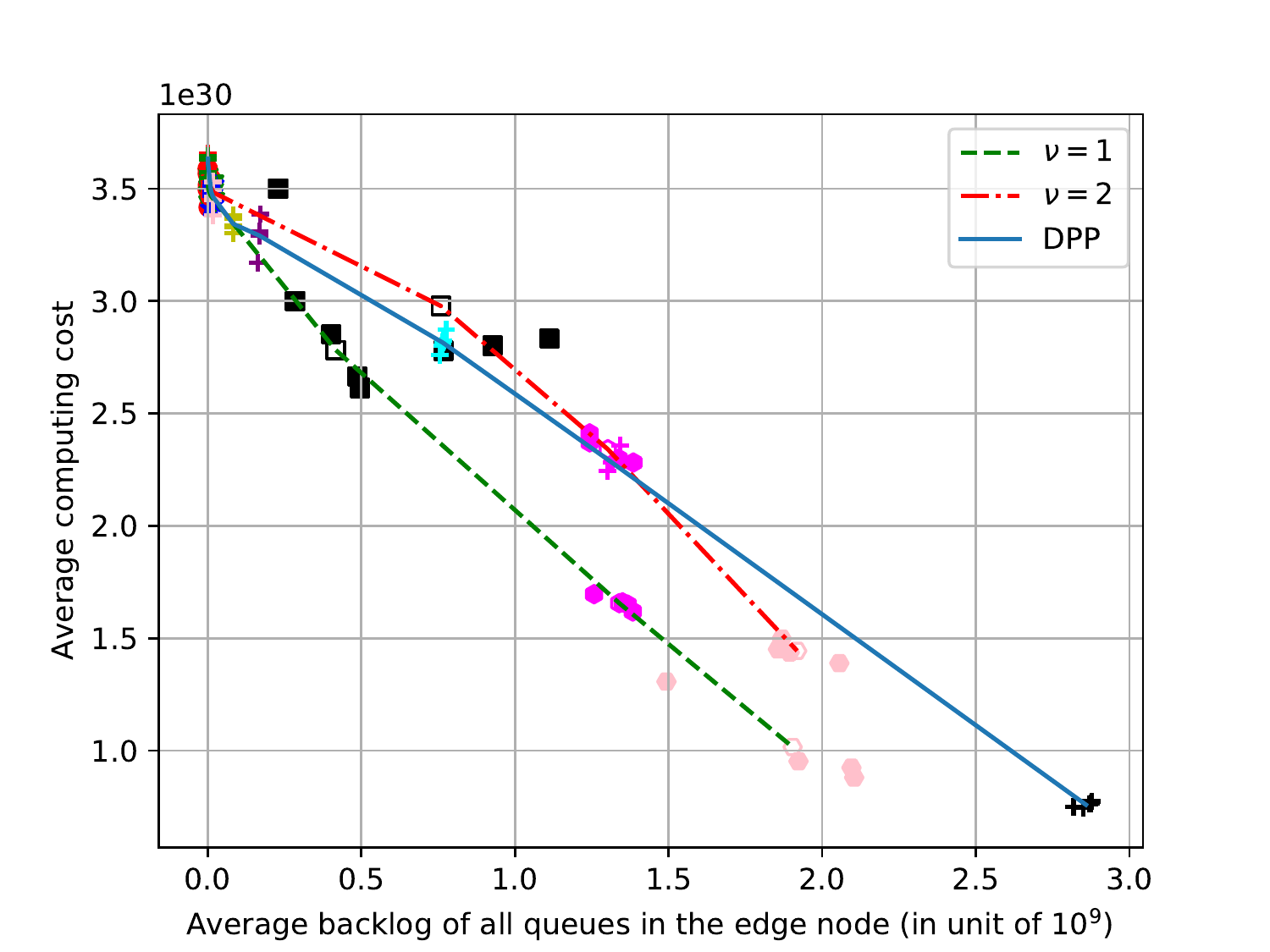}
    \caption{Average episode penalty  versus average episode queue length}
    \label{fig:performanceComparison}
    \end{figure}

With the verification of our DRL algorithm's convergence, we compared the DRL algorithm with the DPP algorithm conventionally used for the Lyapunov framework.   
Fig. \ref{fig:performanceComparison} shows the trade-off performance between the penalty cost and the average queue length of the two methods. For the DRL method, we assumed that the policy has converged after 6$M$ and 20$M$ time steps for $\nu=1$ and $\nu=2$, respectively, based on the result in Fig.   \ref{fig:learningCurves_step11}, and picked this converged policy as our execution policy.  With the execution policy, we ran several episodes for each  $V$ and computed the average episode penalty cost $\frac{1}{T}\sum_{t=0}^{T-1}[C_E(t)+C_C(t)]$ and the average episode queue length $\frac{1}{T}\sum_{t=0}^{T-1}\sum_{i=1}^N q_i(t)$. We plotted the points in the 2-D plane of the average episode queue length and the average episode penalty cost by sweeping $V$. The result is shown in Fig. 
\ref{fig:performanceComparison}. 
Each  line in Fig. \ref{fig:performanceComparison} is the connecting line through the mean value of multiple episodes for each $V$ for each algorithm. For the DPP algorithm, multiple episodes with the same length $T=5000$ were tested for each $V'$ and the trade-off curve was drawn. The weighting factors $V$ and $V'$ in \eqref{eq:SimReward} and \eqref{eq:DPPsimReward} were separately swept.  It is seen that the DRL approach with $\nu=2$ shows a similar trade-off performance to that of the DPP algorithm, and the DRL approach with $\nu=1$ shows a better trade-off performance than the DPP algorithm. This is because the stability related part of the reward directly becomes the queue length  when $\nu=1$.
    \begin{figure}[t]
    \centering
    \subfigure[Avg. bits: $2.21\cdot10^5$]{
        \centering
        \includegraphics[width=0.45\linewidth]{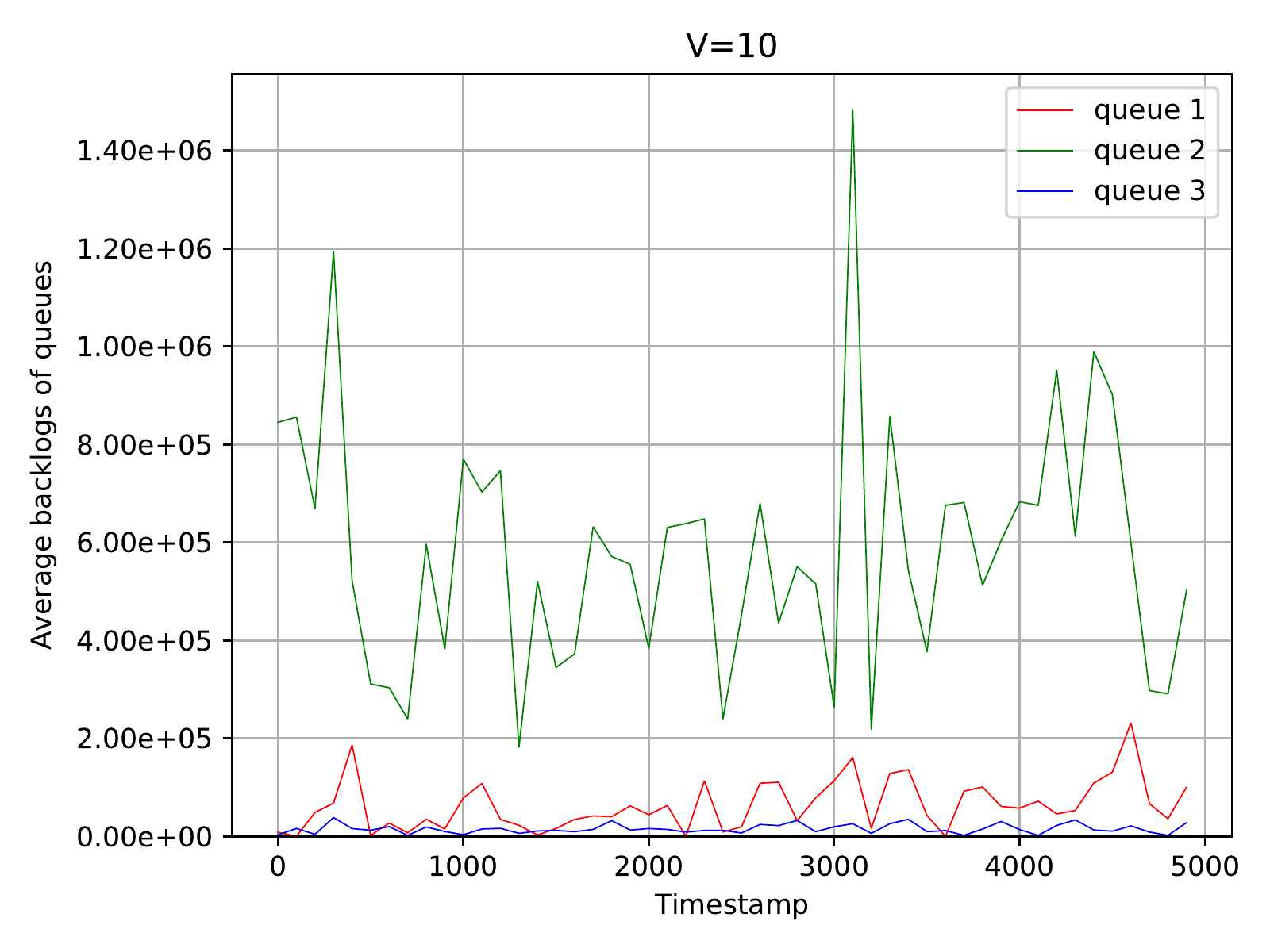}
    }\hfil
    \subfigure[Avg. bits: $2.38\cdot10^6$]{
        \centering
        \includegraphics[width=0.45\linewidth]{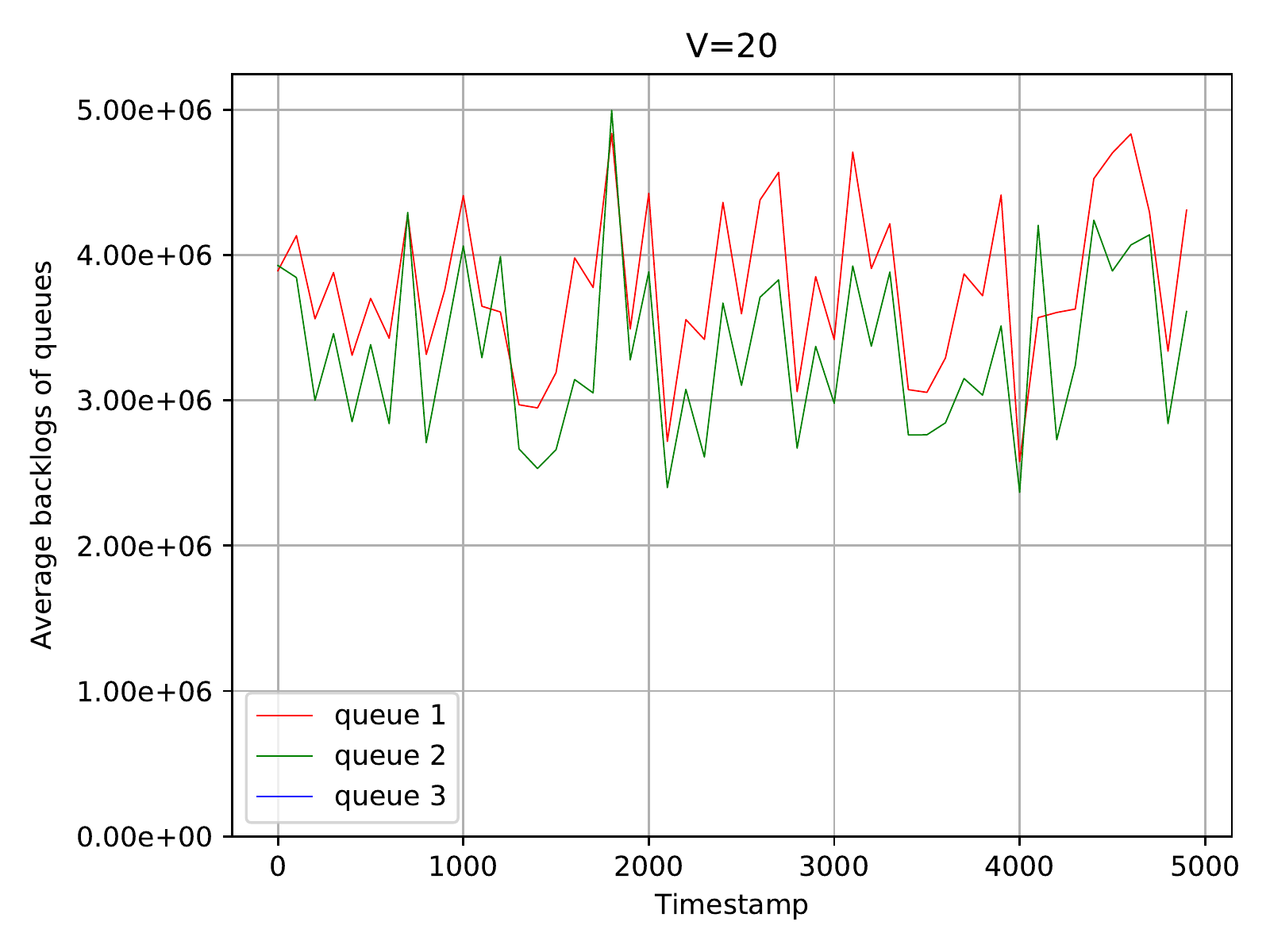}
    }\hfil
    \subfigure[Avg. bits: $1.40\cdot10^7$]{
        \centering
        \includegraphics[width=0.45\linewidth]{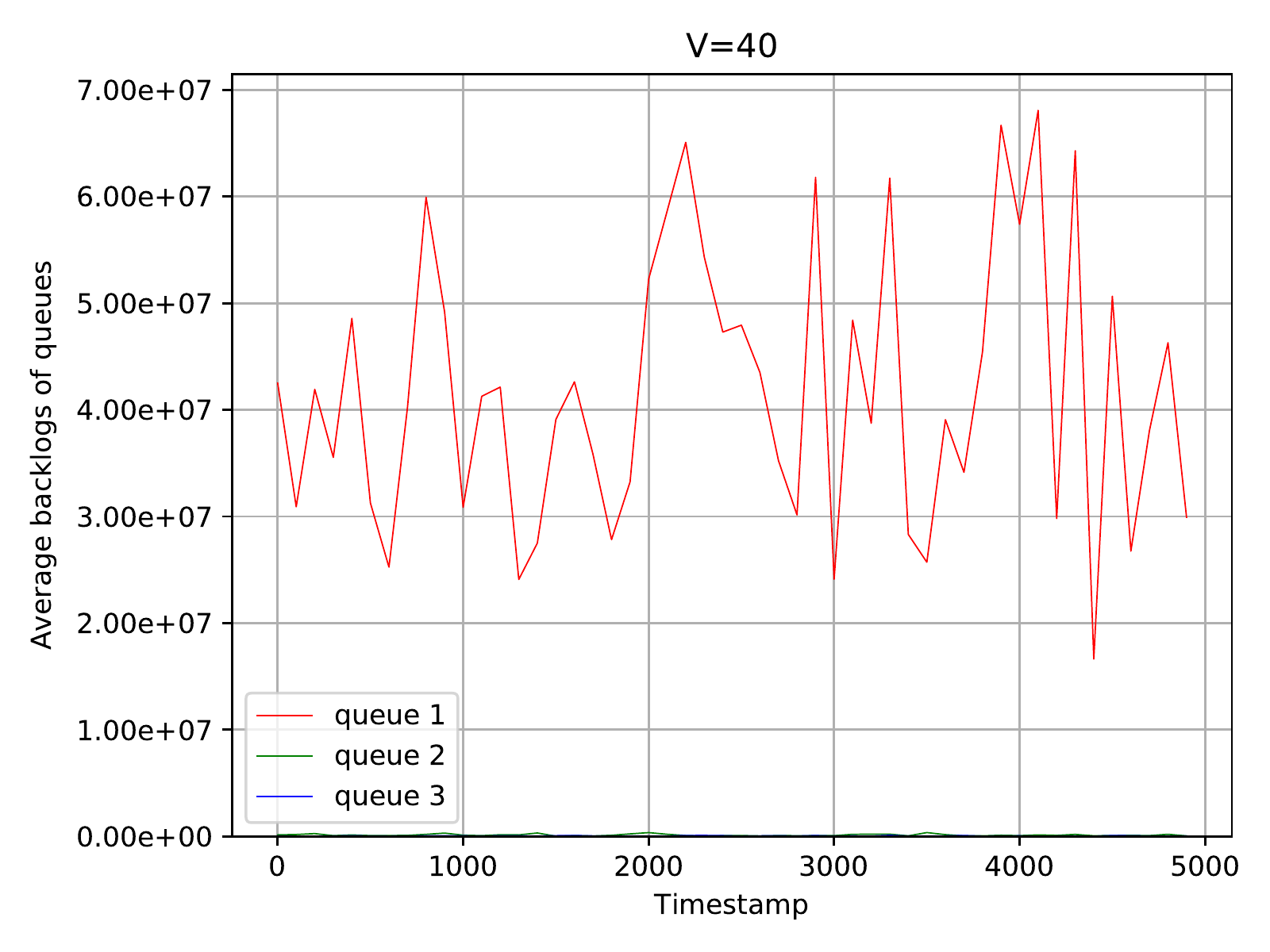}
    }\hfil
    \subfigure[Avg. bits: $4.96\cdot10^8$]{
        \centering
        \includegraphics[width=0.45\linewidth]{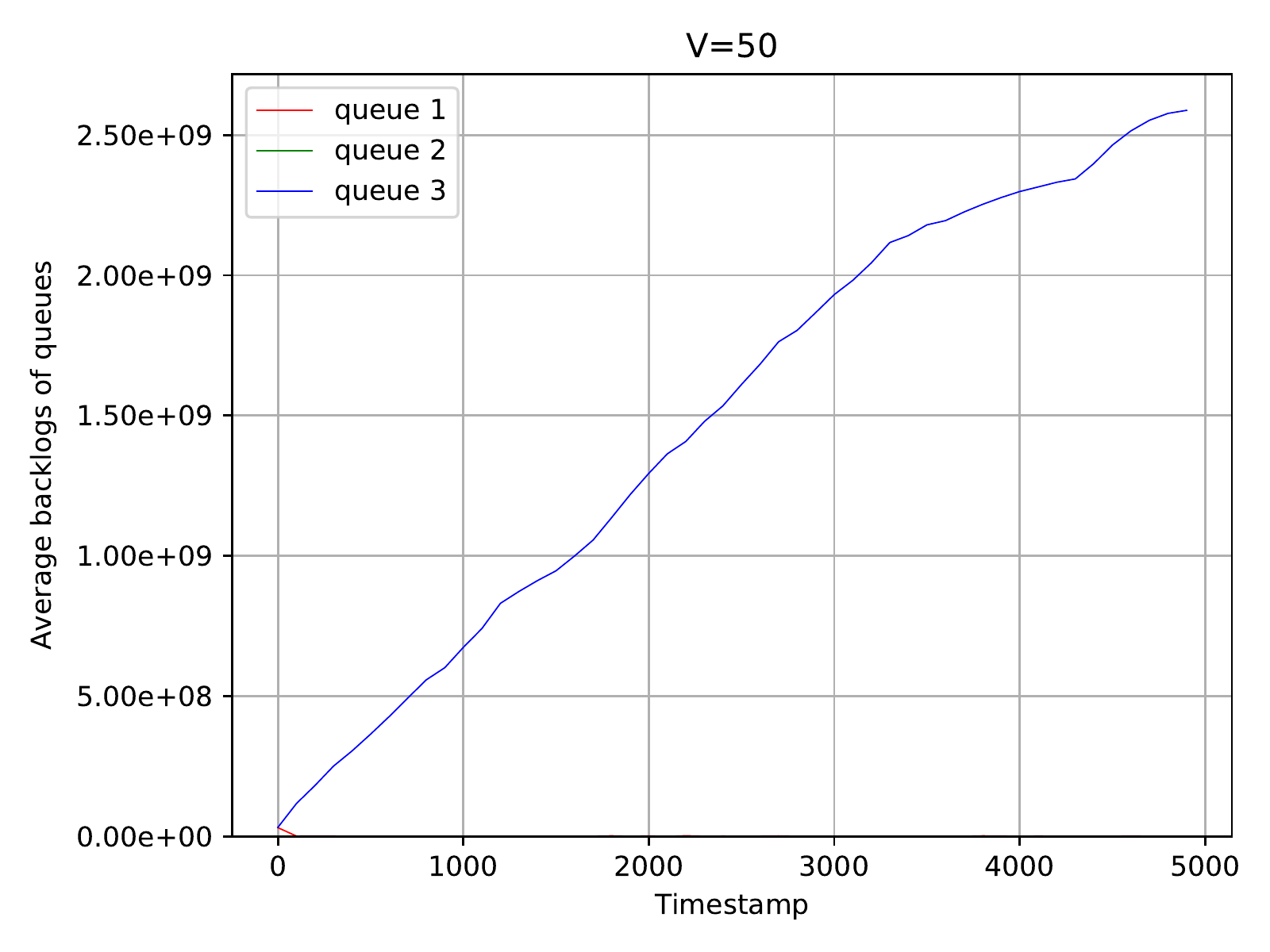}
    }
    \caption{Queue length $\sum_{i=1}^N q_i(t)$ over time: $\nu=1$}
    \label{fig:q_for_even_drl_nu1}
    \end{figure}
    \begin{figure}[h]
    \centering
    \subfigure[Avg. bits: $9.75\cdot10^6$]{
        \centering
    \includegraphics[width=0.45\linewidth]{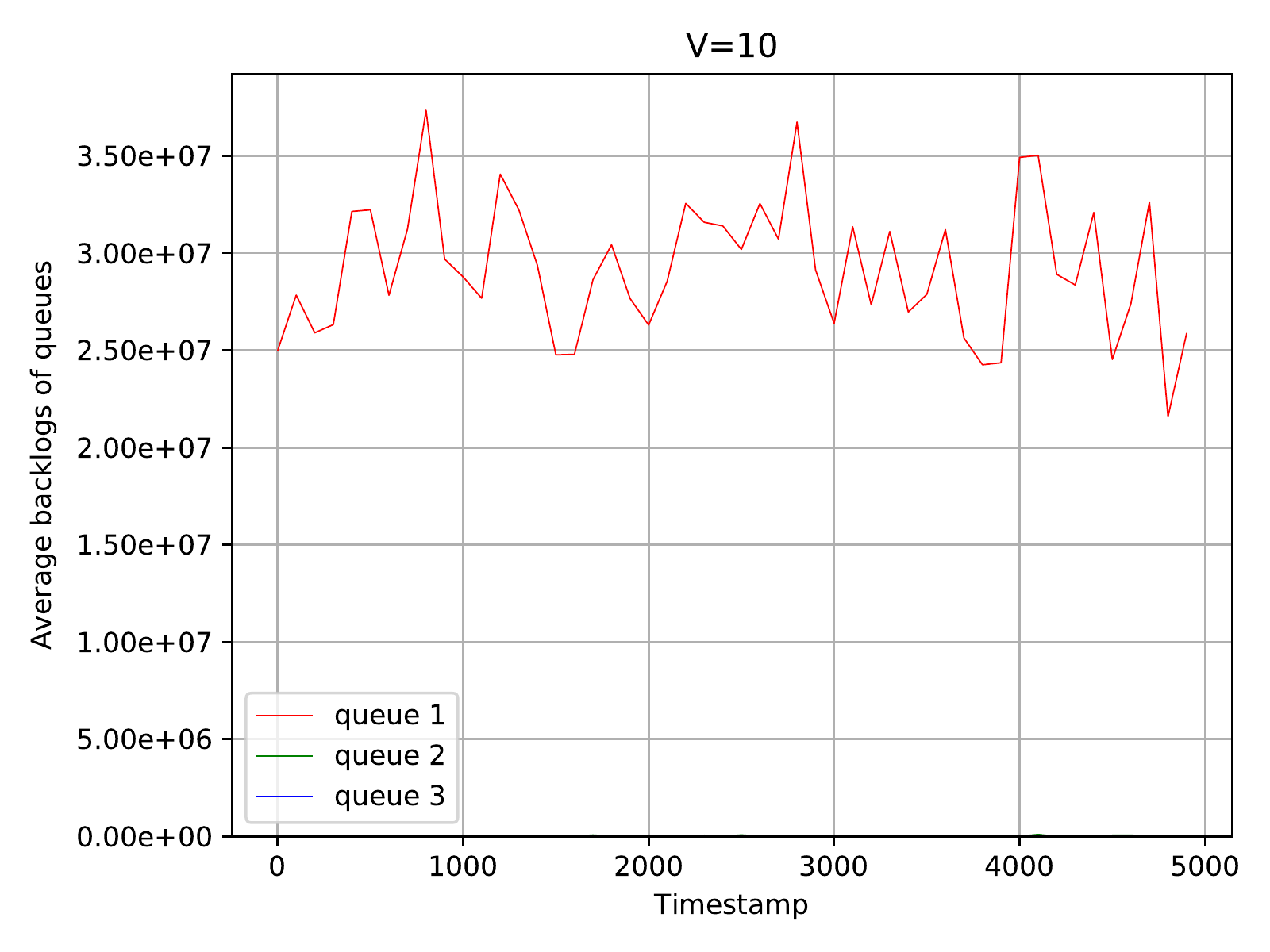}
    }\hfil
    \subfigure[Avg. bits: $2.06\cdot10^7$]{
        \centering
        \includegraphics[width=0.45\linewidth]{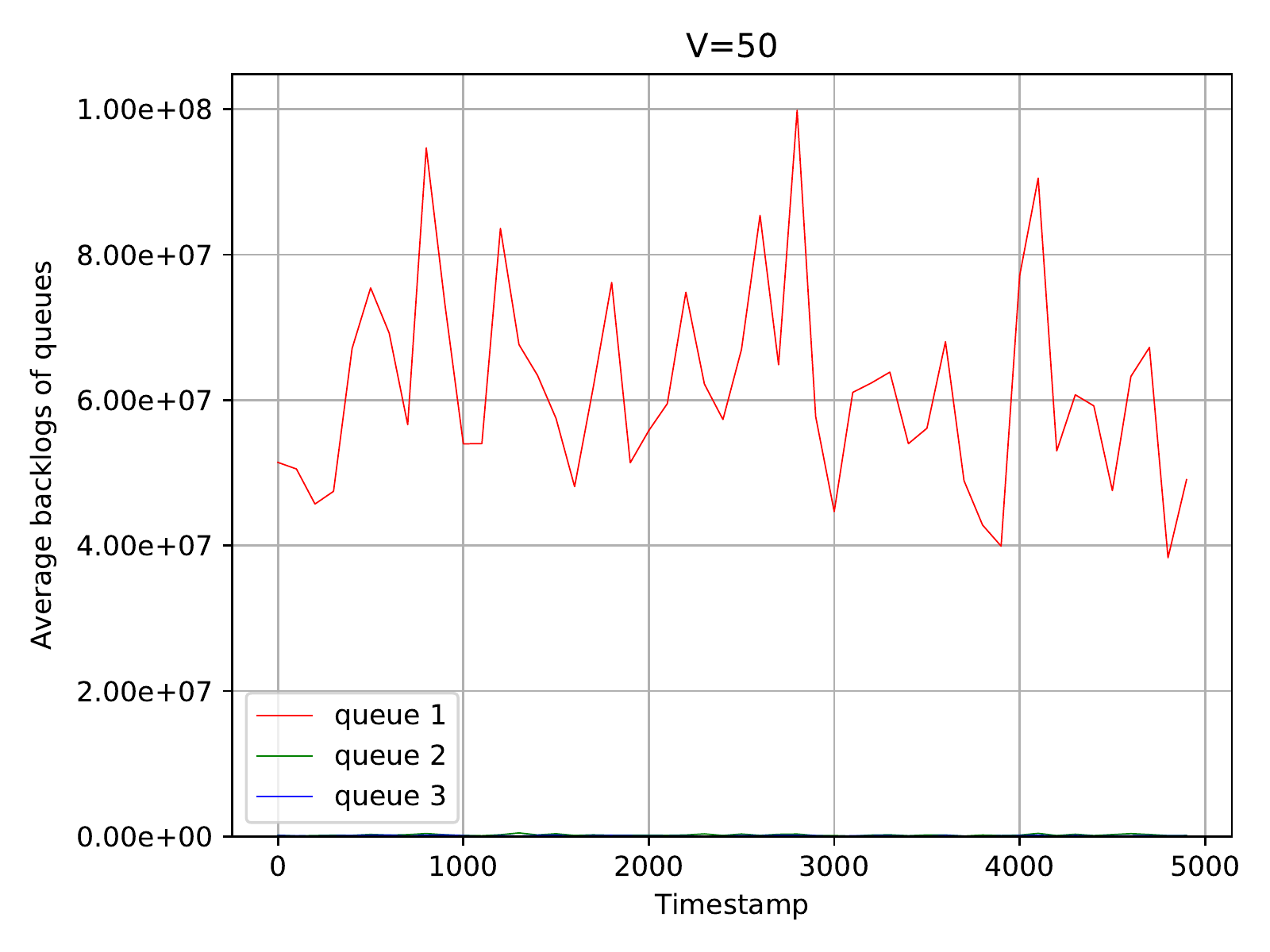}
    }\hfil
    \subfigure[Avg. bits: $1.11\cdot10^9$]{
        \centering
        \includegraphics[width=0.45\linewidth]{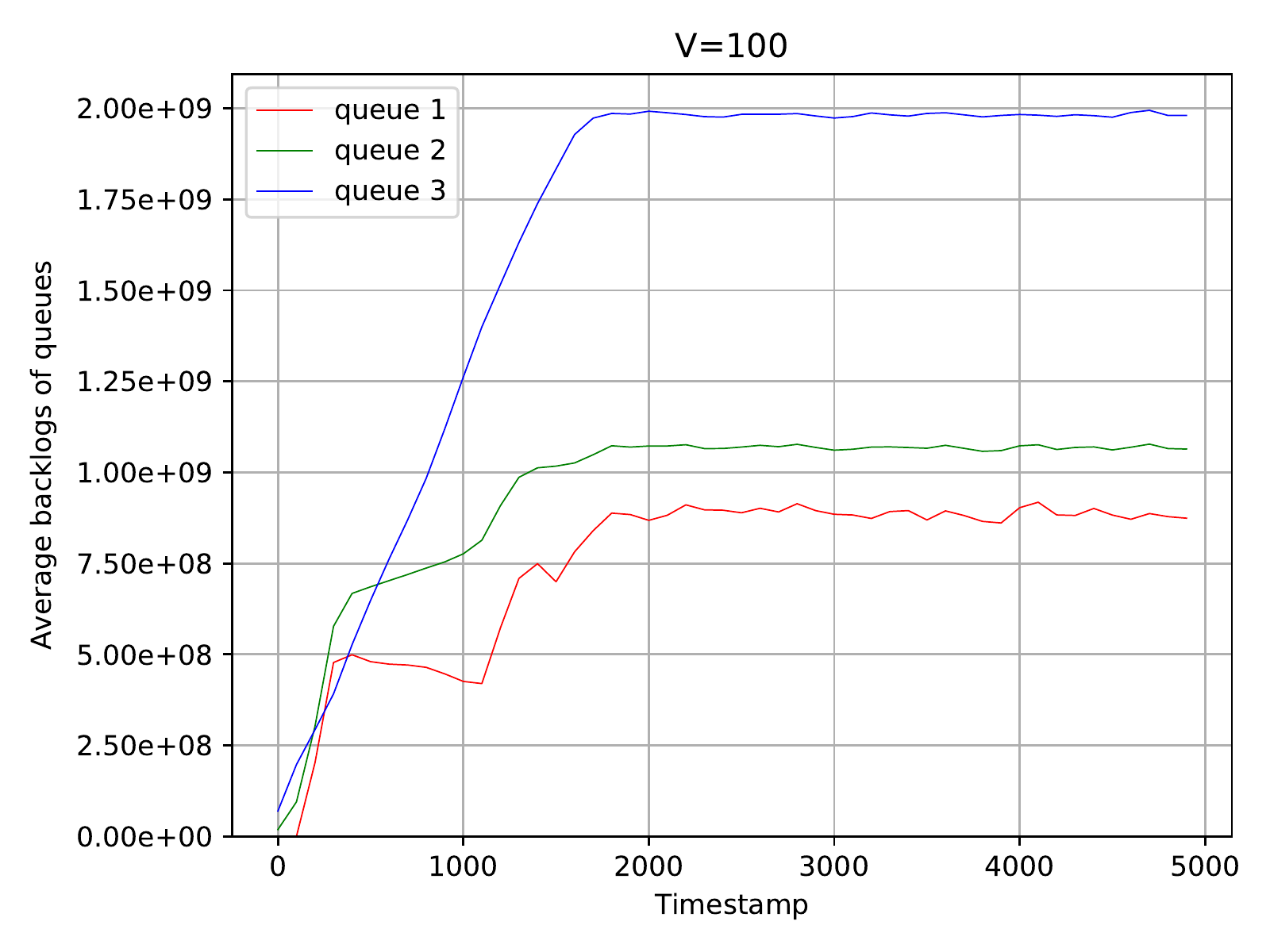}
    }\hfil
    \subfigure[Avg. bits: $1.33\cdot10^9$]{
        \centering
        \includegraphics[width=0.45\linewidth]{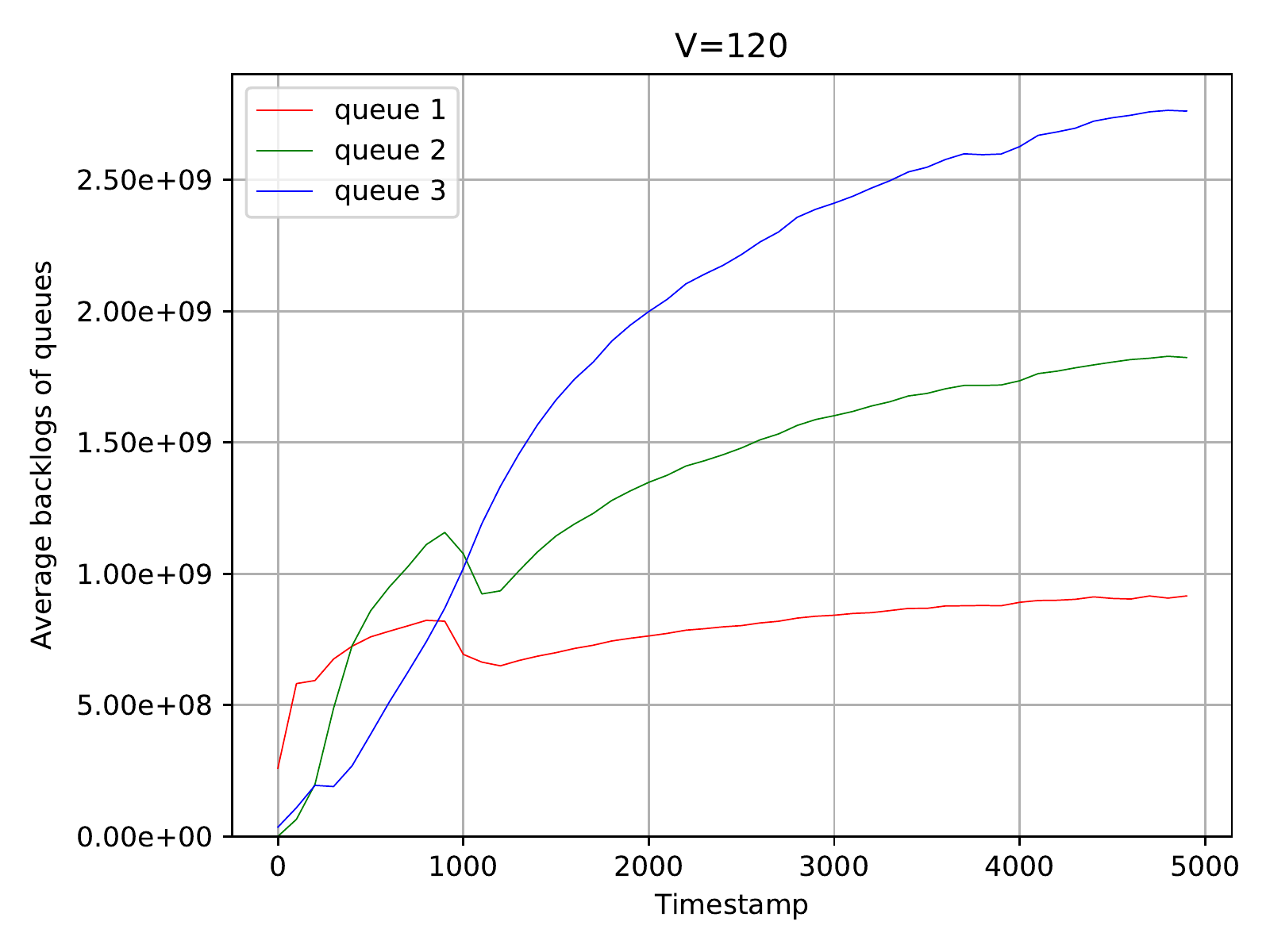}
    }
    \caption{Queue length $\sum_{i=1}^N q_i(t)$ over time: $\nu=2$}
    \label{fig:q_for_even_drl_nu2}
    \end{figure}
Then, we checked the actual queue length evolution with initial zero queue length for the execution policy, and the result is shown in Figs.  \ref{fig:q_for_even_drl_nu1} and
    \ref{fig:q_for_even_drl_nu2} 
    for $\nu=1$ and $\nu=2$, respectively. 
    It is seen that  the queues are stabilized up to the average length $\sim 10^8$ for $\nu=1$ and up to the average length $\sim 10^9$ for $\nu=2$.
    However, beyond a certain value of $V$, the penalty cost becomes dominant and the agent learns a policy that focus on the penalty cost reduction while sacrificing the queue stability.  Hence,  the upper left region in      Fig. \ref{fig:performanceComparison} is the desired operating region with queue stability.

\subsection{General reward function: A discontinuous function case}
\label{subsubsec:discontinuousRewFunc}

In the previous experimental example, it is observed  that in the queue-stabilizing operating region the performance of DRL-SAC and the DPP performance is more or less the same, as seen in Fig. \ref{fig:performanceComparison}. This is because the DPP algorithm yields a solution with performance within a constant gap from the optimal value due to the Lyapunov optimization theorem.  The effectiveness of the proposed DRL approach is its versatility for general reward functions in addition to the fact that optimization is not required once the policy is trained. 
Note that the DPP algorithm requires solving a constrained nonconvex optimization for each time step in the general reward function case.  Although such constrained nonconvex optimization can be approached by several methods such as  successive convex approximation (SCA)  \cite{Scutari:book} as we did in Section \ref{subsec:simDPPcomp}. However, such methods requires certain properties on the reward function such as continuity, differentiability, etc. and it may be difficult to apply them to general reward functions such as reward given by a table.   In order to see the generality of the DRL approach to the Lyapunov optimization, we considered a more complicated penalty function. We considered the same reward $C_E(t)$ given by 
\eqref{eq:CEtFormula} but  instead of \eqref{eq:CctFormula},  $C_C(t)$ was given by a scaled version of  the number of CPU cores at the cloud node required to process the offloaded tasks under the assumption that each CPU core was fully loaded with it maximum clock rate $4$ GHz before the next core was assigned. This $C_C(t)$ is a discontinuous function of the amount of the offloaded task bits. All other set up was the same as that in Section \ref{subsec:simDPPcomp}. With this penalty function, it is observed that the DPP algorithm based on SQP failed but the DRL-SAC with the proposed state and reward function successfully learned a policy. 
Fig. \ref{fig:learningCurves_step} shows the learning curves of DRL-SAC in this case for different values of $V$ with $\nu=1$. (The plot was obtained in the same way as that used for   Fig. \ref{fig:learningCurves_step11}.) Fig.  \ref{fig:tradeoffperformance_disconti} shows the corresponding DRL-SAC trade-off performance between the average episode penalty and the average episode queue length, and Fig.  \ref{fig:q_for_step} shows the queue length over time for one episode for the trained policy in the execution phase. 
Thus,  the DRL approach operates properly in a more complicated penalty function for which the DPP algorithm may fail.

    \begin{figure}[t]
        \subfigure{
        \centering
        \includegraphics[width=\linewidth]{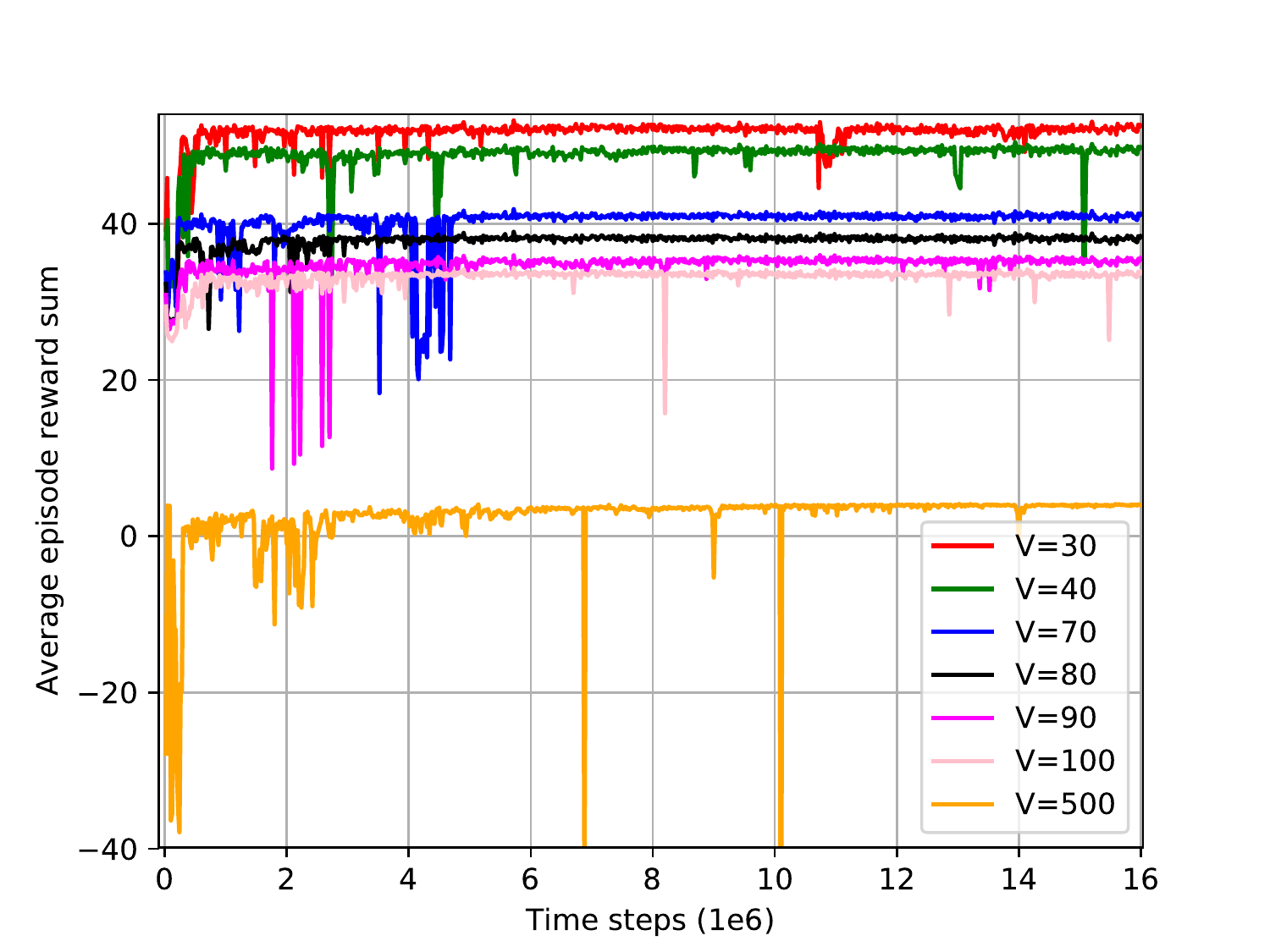}
        }
    \caption{Learning curves of DRL-SAC for the discontinuous $C_C(t)$ function}
    \label{fig:learningCurves_step}
    \end{figure}

    \begin{figure}[!t]
        \subfigure{
        \centering
        \includegraphics[width=\linewidth]{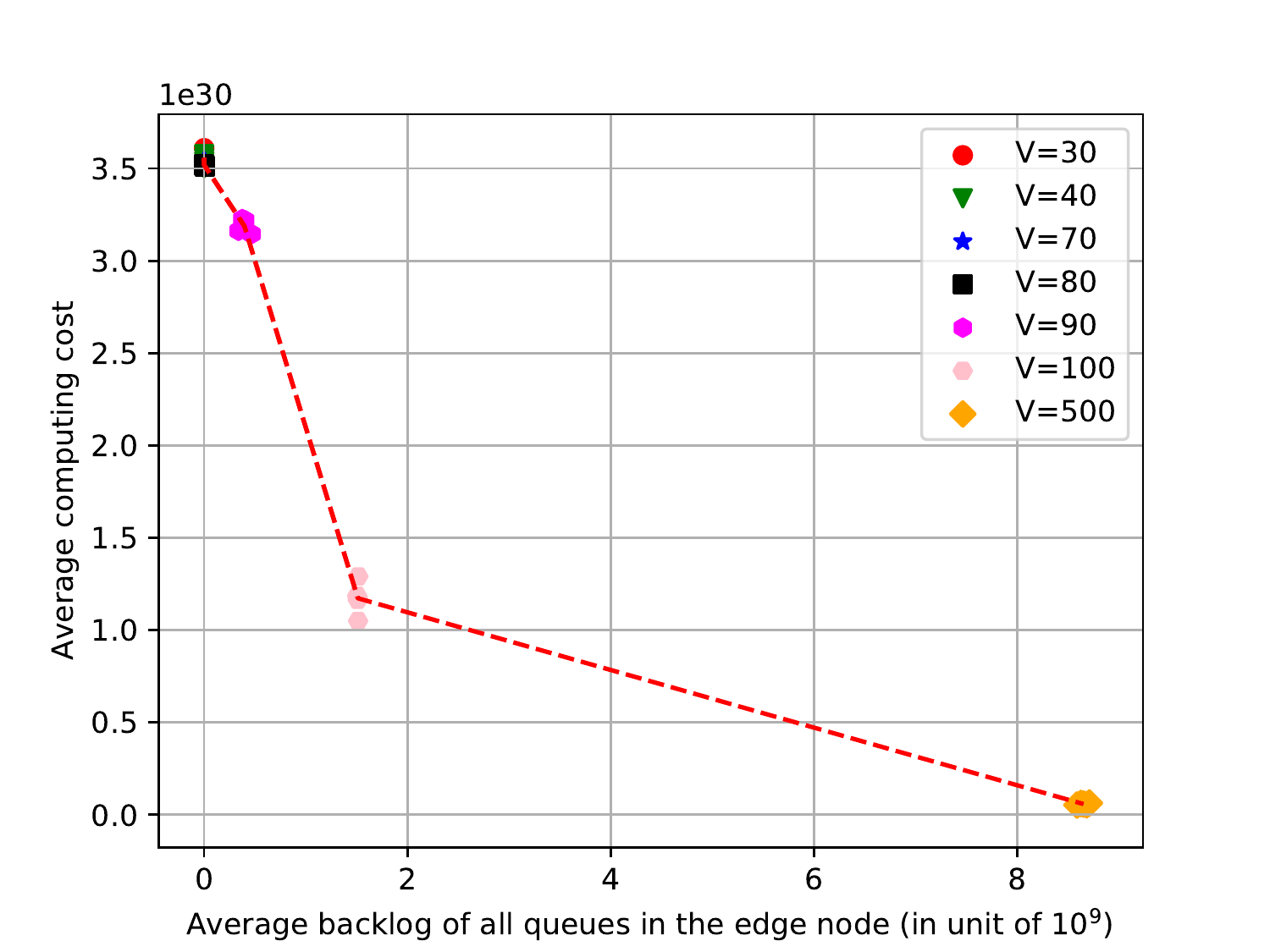}
        }
    \caption{Average episode penalty  versus average episode queue length}
    \label{fig:tradeoffperformance_disconti}
    \end{figure}

    \begin{figure}
    \centering
    \subfigure[]{
        \centering
        \includegraphics[width=0.45\linewidth]{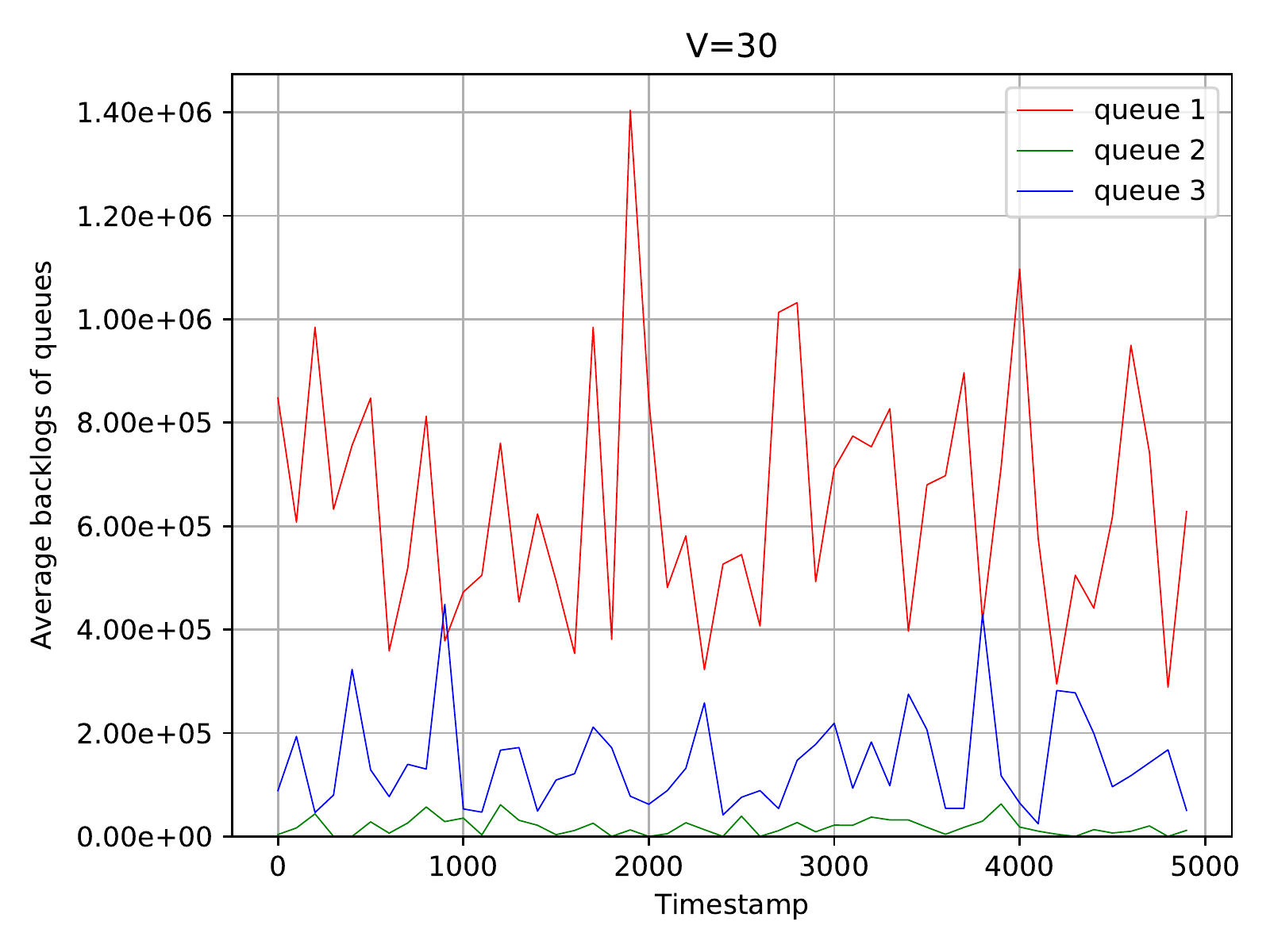}
    }\hfil
    \subfigure[]{
        \centering
        \includegraphics[width=0.45\linewidth]{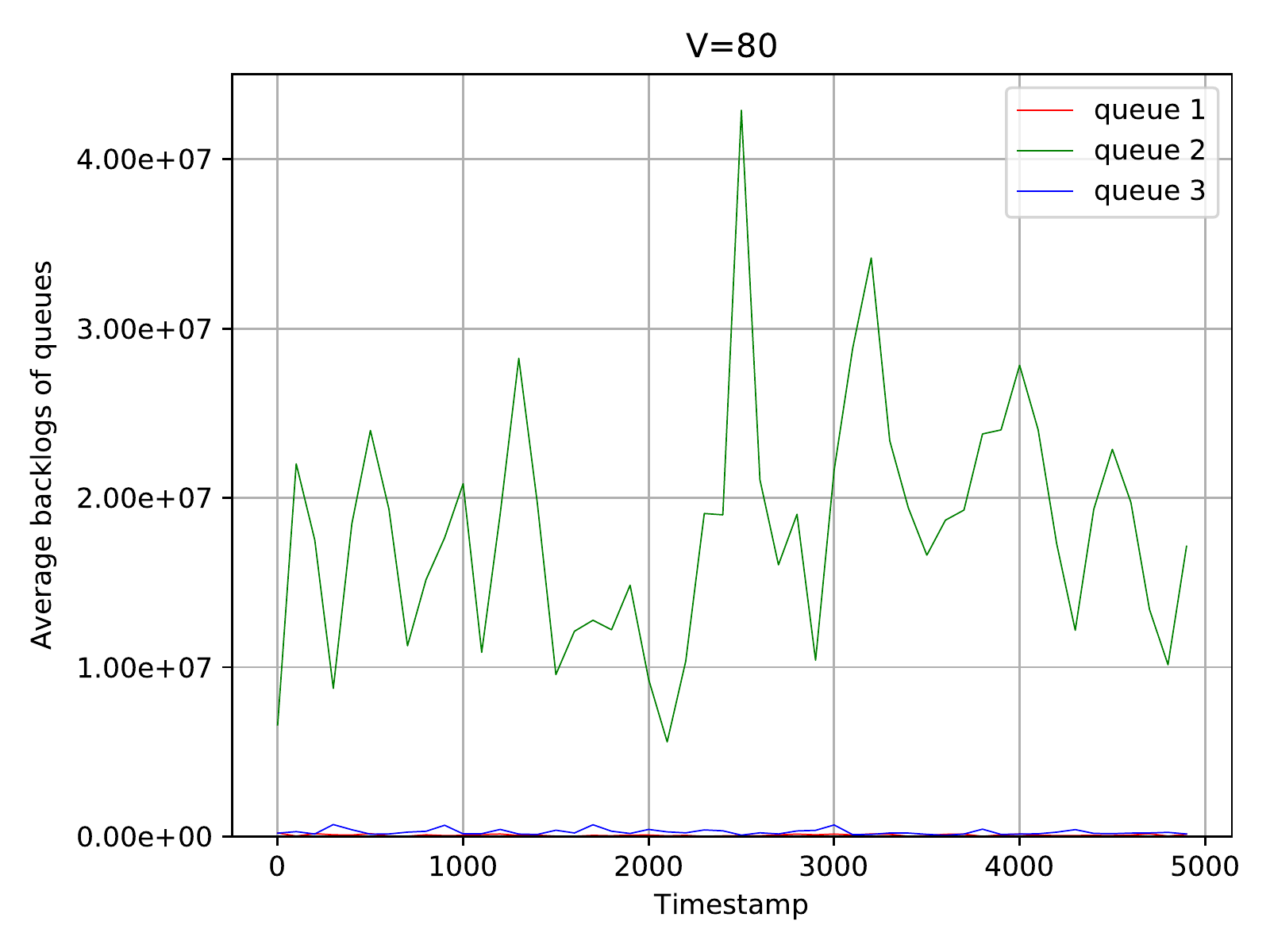}
    }\hfil
    \subfigure[]{
        \centering
        \includegraphics[width=0.45\linewidth]{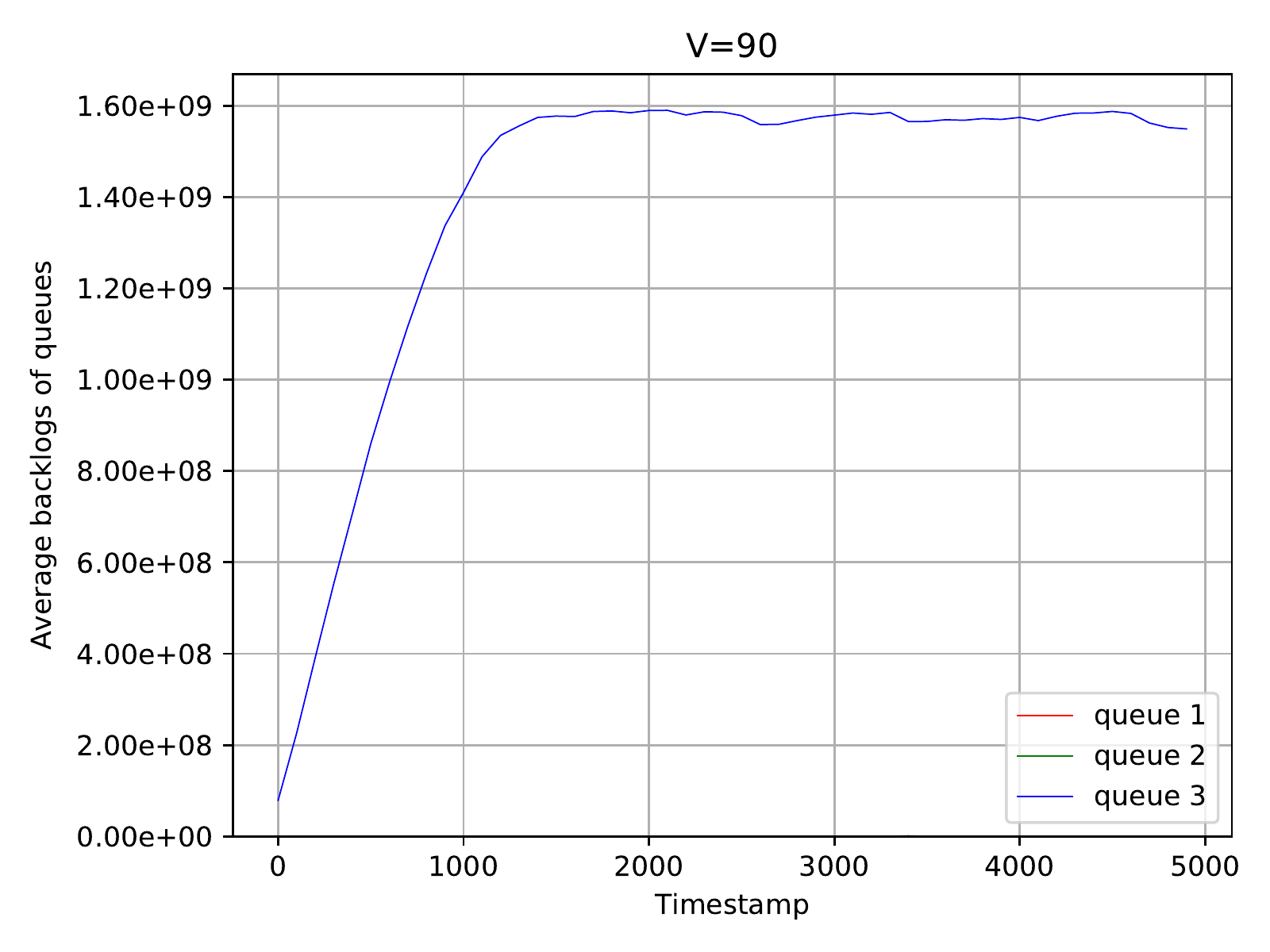}
    }\hfil
    \subfigure[]{
        \centering
        \includegraphics[width=0.45\linewidth]{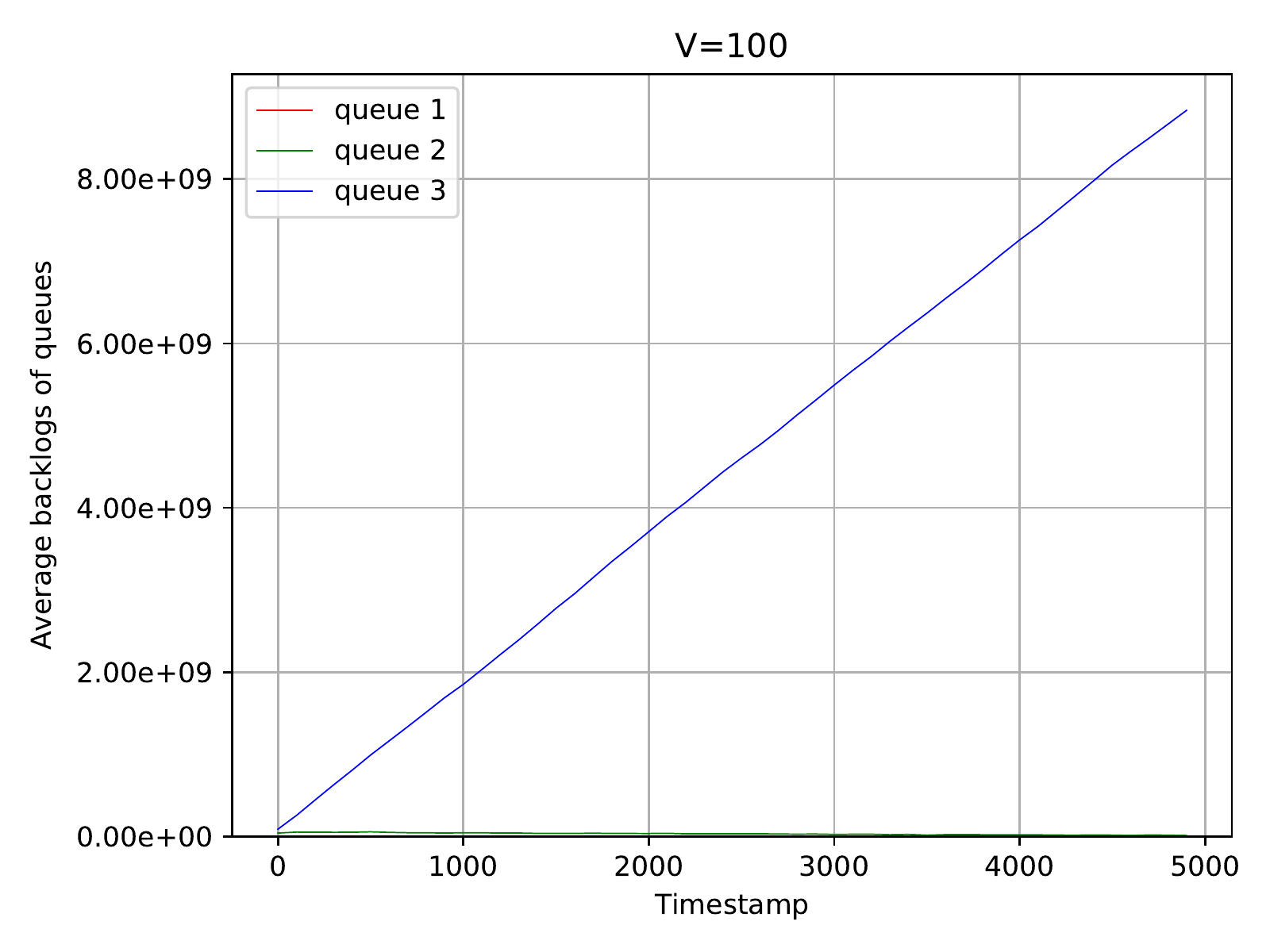}
        }
    \caption{Queue length $\sum_{i=1}^N q_i(t)$ over time: (a) V=30, (b) V=80, (c) V=90, and (d) V=100}
    \label{fig:q_for_step}
    \end{figure}

\subsection{Operation in Higher Action 
Dimensions}

    \begin{table}[h]
    \centering
    \caption{Application parameters}
        \begin{tabular}{|l|c|c|c|}
        \hline
        \rowcolor{lightgray} 
       \begin{tabular}[c]{@{}l@{}}\textbf{Application} \\ \textbf{names}\end{tabular}
       & $w_i$  &  \textbf{Distribution of $d_i$}  &$\lambda_i$ \\ \hline
        \begin{tabular}[c]{@{}l@{}}
        \textbf{Speech}\\ \textbf{recognition}
        \end{tabular} & 10435  & $N_T(170,130,40,300)$(KB)  & 0.5 \\ \hline
        \textbf{NLP}  & 25346 & $N_T(52,48,4,100)$(KB)   & 0.8  \\ \hline
        \begin{tabular}[c]{@{}l@{}}\textbf{Face}\\ \textbf{Recognition}\end{tabular}
         & 45043   & $N_T(55,45,10,100)$(KB) &0.4  \\ \hline
        \textbf{Searching}	&	8405 &	$N_T(51,24.5,2,100)$ (byte) & 10 \\ \hline
        \textbf{Translation}  &	34252 &	$N_T(2501,1249.5,2,5000)$ (byte) &	1 \\ \hline
        \textbf{3d game}  &54633&	$N_T(1.55,0.725,0.1,3)$ (MB) &	0.1 	\\ \hline
        \textbf{VR}   	&40305 &	$N_T(1.55,0.725,0.1,3)$ (MB)  &	0.1 \\ \hline
        \textbf{AR}    &34532 &	$N_T(1.55,0.725,0.1,3)$ (MB)&	0.1 	\\ \hline
        \end{tabular}
    \label{tab:8_app_info}
    \end{table}
    
In the previous experiments, we considered the case of $N=3$, i.e., three queues and the action dimension  was $2N+2=8$. In order to check the operability of the DRL approach in a higher dimensional case, we considered the case of $N=8$. In this case, the action dimension\footnote{In the Mujoco robot simulator for RL algorithm test, the Humanoid task is known to hae high action dimensions given by 17 \cite{Mujoco}.} was $2N+2=18$.
The parameters of the eight application types that we considered are shown in Table  \ref{tab:8_app_info}.  Other parameters and setup were the same as those in the case of $N=3$ and $C_C(t)$ was the discontinuous function used in Section \ref{subsubsec:discontinuousRewFunc}.
From Table  \ref{tab:8_app_info} the average total arrival rate in terms of CPU cycles and task bits per second were 193GHz and  5.14 kbps. Hence, the set up was feasible to control. Fig. \ref{fig:q_for_8apps}(a) shows the corresponding learning curve and Fig. \ref{fig:q_for_8apps}(b) shows the queue length $\sum_{i=1}^N q_i(t)$ over time for one episode for the trained policy in the execution phase. It is seen that even in this case the DRL-based approach properly works.

\begin{figure}[!th]
    \subfigure[]{
    \centering
    \includegraphics[width=0.9\linewidth]{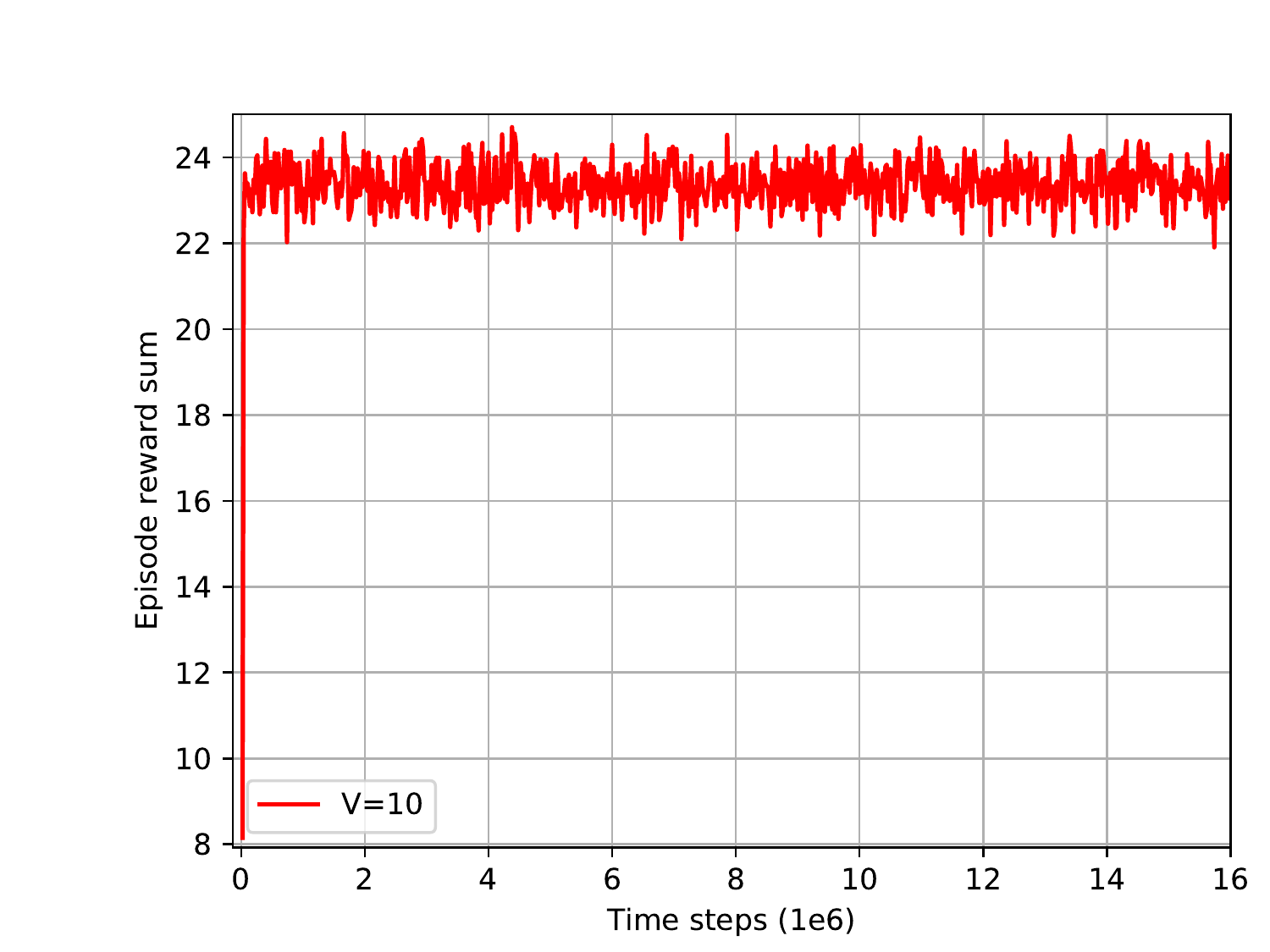}     }
    \subfigure[]{
    \centering
    \includegraphics[width=0.9\linewidth]{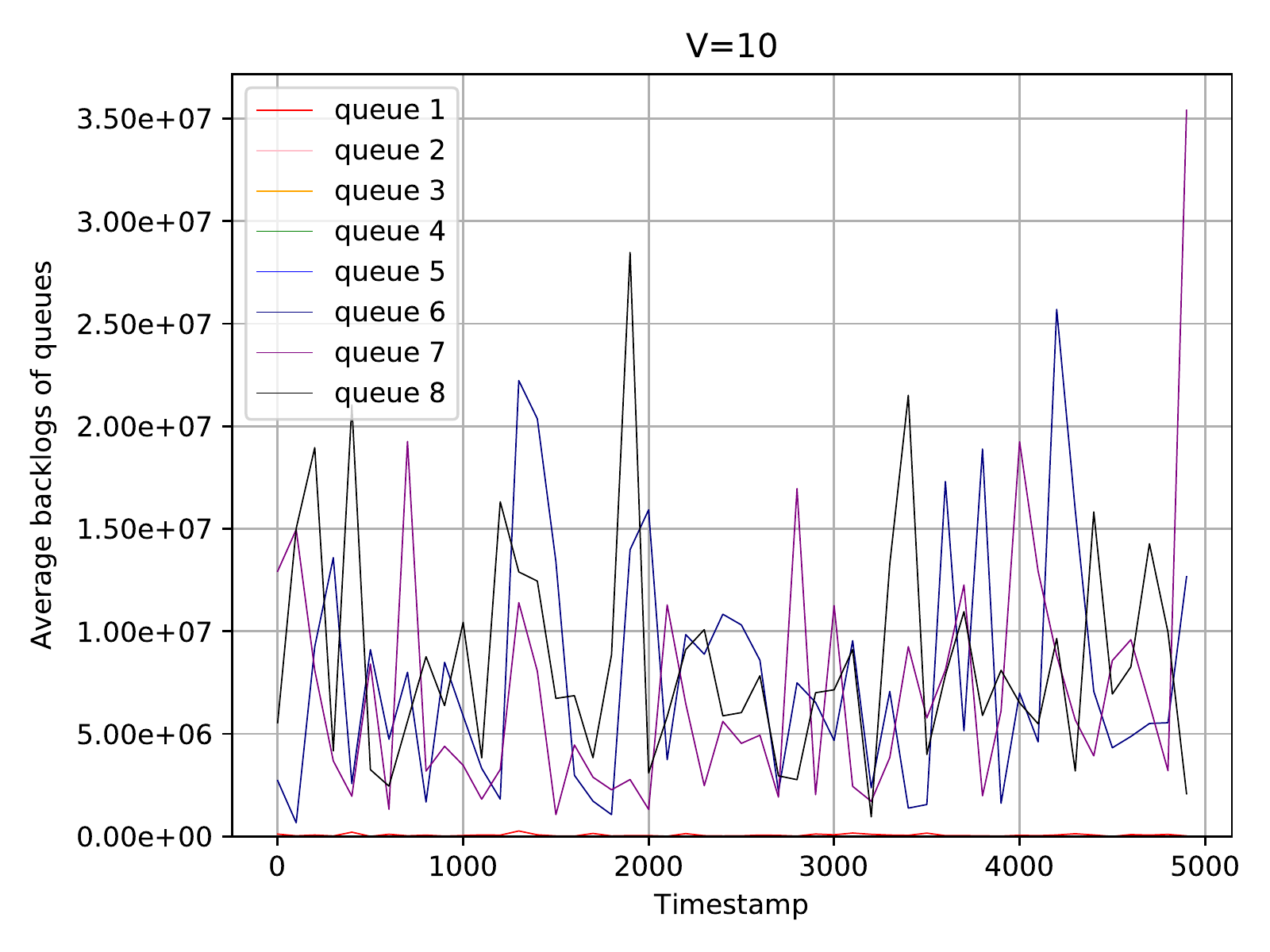}
    }
 \caption{Higher Dimension Case ($\nu=1$): (a) learning curve and (b) queue length over time in the execution phase}
\label{fig:q_for_8apps}
\end{figure}

\section{Conclusion} \label{sec:conclusion}

In this papper, we have considered a DRL-based approach to the Lyapunov optimization that minimizes the time-average penalty cost while maintaining queue stability. We have proposed a proper construction of state and action spaces and a class of reward functions. We have derived  a condition for the reward function of RL for queue stability and have provided a discounted form of the reward for practical RL. With the proposed state and action spaces and the reward function, the DRL approach successfully learns  a policy  minimizing the penalty cost while maintaining queue stability. The proposed DRL-based approach to Lyapunov optimization  does not required complicated optimization at each time step and can operate with general non-convex and discontinuous penalty functions. Thus, it provides an alternative to the conventional DPP algorithm  to the Lyapunov optimization.


\phantomsection
\addcontentsline{toc}{chapter}{Bibliography}
\bibliographystyle{IEEEtran}
\bibliography{201105}

\end{document}